\documentclass[12pt]{article}

\usepackage{epsf}
\usepackage{fancyhdr}
\usepackage{graphics}
\usepackage{graphicx}

\usepackage{blurb_style}

\usepackage{color}

\usepackage{amsfonts}
\usepackage{amsmath}
\usepackage{amssymb,bbm}

\usepackage{uri}

\usepackage[colorlinks, allcolors=blue]{hyperref}

\usepackage{booktabs}

\usepackage{caption}
\usepackage{subcaption}

\usepackage{mathtools}

\numberwithin{equation}{section}

\usepackage{color}
\usepackage{comment}
\usepackage{pdfpages}
\usepackage{dcolumn} % Align table columns on decimal point

\usepackage[authoryear, round]{natbib}

\begin{document}

% "Title of the paper"
\title{Using Simulation to Analyze \\ Interrupted Time Series Designs\thanks{I gratefully acknowledge the productive and stimulating collaboration on this project with MDRC, in particular Chloe Anderson, Brit Henderson, Cindy Redcross, and Erin Valentine. This work is in part generously funded by Arnold Ventures. Further thanks to the useful feedback and commentary from Howard Bloom, Rich Hendra, and the Miratrix CARES Lab. Direct correspondence to \href{mailto:lmiratrix@gse.harvard.edu}{lmiratrix@gse.harvard.edu}}}

\author{Luke Miratrix\\Harvard Graduate School of Education}

\maketitle

\begin{abstract}
We are sometimes forced to use the Interrupted Time Series (ITS) design as an identification strategy for potential policy change, such as when we only have a single treated unit and no comparable controls. For example, with recent county- and state-wide criminal justice reform efforts, where judicial bodies have changed bail setting practices for everyone in their jurisdiction in order to reduce rates of pre-trial detention while maintaining court order and public safety, we have no natural comparison group other than the past.
In these contexts, it is imperative to model pre-policy trends with a light touch, allowing for structures such as autoregressive departures from any pre-existing trend, in order to accurately and realistically assess the statistical uncertainty of our projections (beyond the stringent assumptions necessary for the subsequent causal inferences).
To tackle this problem we provide a methodological approach rooted in commonly understood and used modeling approaches that better captures uncertainty.
We quantify uncertainty with simulation, generating a distribution of plausible counterfactual trajectories to compare to the observed; this approach naturally allows for incorporating seasonality and other time varying covariates, and provides confidence intervals along with point estimates for the potential impacts of policy change. 
We find simulation provides a natural framework to capture and show uncertainty in the ITS designs. It also allows for easy extensions such as nonparametric smoothing in order to handle multiple post-policy time points or more structural models to account for seasonality.
\end{abstract}


\section{Introduction}

Interrupted Time Series (ITS) can occur when a governance body of some area, e.g. a county, implements a policy change at a particular point in time.
The researcher is able to observe regular measures of some outcome of interest both for several time points before such a change as well as after.
The research question is then whether there is any evidence that the policy has changed the course of the unit of interest.
%In particular, at the point of policy change, we might naturally expect various outcomes of interest to change, deviating from what they historically were.
For example, if a school undergoes a massive reorganization we would expect that measures of student wellbeing, such as rate of college-going or rate of graduation, might change over time, and change differently than they would have absent the reorganization.
We might not, however, expect any treatment impact right at the time of the policy change as it may take time for the policy to become fully implemented, and for the consequences of the policy to be felt, making  Regression Discontinuity Designs particularly inappropriate.
What makes these contexts particularly challenging is frequently there is only a single unit that received the policy change and/or no reasonable comparison units that did not receive such a change.

One area where we see these kind of reform efforts are in modern criminal justice reform, in particular pretrial reform.
Currently, in the U.S., hundreds of thousands of people are incarcerated in local jails on any given day as they await resolution of their criminal case.
These people have not been convicted, but are nonetheless incarcerated because, generally, they cannot afford to post monetary bail to secure their release \citep{zeng2018jail}.
Several jurisdictions have sought to improve these judicial systems by attempting to build procedures to increase the rate of release for ``low-risk'' defendants.
%ideally, this would reduce negative impacts on these defendants as well as reduce load on jails.
One general category of such reforms use risk assessment tools in early court proceedings, providing judges with information about the risk of a defendant as measured by various characteristics such as previous criminal history in order to improve judicial decision-making regarding what types of supervision or restrictions should be placed on defendants awaiting their case resolution.
%Given the comprehensive nature of the reforms, these evaluations did not lend themselves to comparable time series, such as similar regions with similar judicial systems and contexts, making ITS the clearest approach.\footnote{For comparitive

This is the context we use in this work.
We use data from two such reform efforts, one in Mecklenberg County, NC \citep{Redcross:te}, and one in the state of New Jersey \citep{Redcross:nj}.
There are several primary outcomes of interest, of which we examine two: the proportion of arrestees assigned monetary bail, and the total number of warrant arrests made.
% PUT BACK: The methods described here were also applied to other outcomes (e.g.,  rates of new  arrest,  rates of failure to appear for court dates, time spent in jail); see the reports cited above.

Perhaps the most used analytic approach for ITS is to fit a simple linear regression to the data, regressing the outcome of interest onto time and a series of dummy variables for each time point post-policy.
The estimates of these dummy variables then provide impact estimates for each post-policy point.
Unfortunately, even if the underlying linear trend were fundamentally sound, the deviations from trend are likely correlated and this correlation needs to be taken into account.
Not doing so correctly will undermine any estimates of uncertainty by giving overly precise (too small) standard errors.
%But such a model can be easily critiqued; in particular (and especially if we observe our unit relatively frequently such as monthly) we might imagine that the deviations we observe from the overall linear trend are correlated, which would severely.

We propose to account for local dependencies by fitting an autoregressive model with linear trend to the pre-policy data, and then using that model to simulate, using a pseudo-Bayesian approach discussed in \cite{gelman2006data}, a distribution of plausible post-policy trajectories that we would expect if pre-policy trends continued unabated.
By comparing this distribution to the observed post-policy trend, we can estimate impacts and test for the significance of impacts, given the set of rather stringent assumptions necessary for an ITS analysis.
We can also calculate confidence intervals to assess ranges of impact.
This simulation procedure takes into account the uncertainty of the linear model estimate, uncertainty in the measurement of the outcomes, and any autoregressive dependencies in the residuals.
Simulation allows us to easily summarize and visualize heterogeneous impacts post-policy, and in general is an approach that can enrich classic inference \cite{King:vh}.

Simulation also allows for several natural extensions.
First, we can easily incorporate covariates to capture nonlinearities (in particular, seasonal trends).
Second, we can average, or smooth, multiple months of potentially heterogeneous impacts typically found in such evaluations to better capture post-policy impacts in interpretable ways.
This allows testing whether a \emph{group} of post-policy time points differs statistically significantly from what would have occurred in the absence of an intervention.
We provide an R package, \verb|simITS|\footnote{See most recent version at \url{https://github.com/lmiratrix/simITS/}}, that implements the methods discussed along with all routines needed to conduct a full and transparent ITS analysis.

The idea for simulation for assessing uncertainty in these contexts is not new; see, for example, \cite{Zhang:2009de}, who use a parametric bootstrapping approach to assess uncertainty.
Our method is also a parametric simulation approach, but we explicitly include autoregressive dependencies and explicitly simulate post-policy trajectories.
We also discuss the estimands of interest more explicitly.
Similarly, \cite{brodersen_causalimpact} propose a more complex, fully Bayesian time-series approach,  implemented with the \verb|causalImpact| package, that relies on modeling a latent state space.

More broadly, ITS is a generally worse (in terms of strength of evidence) version of \emph{Comparative} Interrupted Time Series (CITS) analyses, where the target treatment series has comparison units that are not treated.
For an overview of CITS, consider \cite{Somers:2013vg} or \cite{Hallberg:2018kp}.
Also see \cite{Jacob:2016jv}, who evaluate the CITS by comparing its findings to those utilizing the more widely-accepted Regression Discontinuity Design.
For a detailed case study with CITS in the context of experimental trials, see \cite{Bloom:2005vc}; this approach has ties to ITS as they fit regressions to the sequence of paired differences.
%Interestingly, many treatments of CITS do not account for autoregressive structure, instead assuming the errors around the modeled linear trends are at least independent, if not also identically distributed (this full assumption is implied by the linear modeling with fixed effects approaches typically used).
%Alternatively, if there are enough units, a CITS analysis can solve the autoregressive issue by simply clustering the standard errors at the unit level, allowing for arbitrary structure.

ITS is also, of course, based on the idea of a time series.
Classic time series methodology, e.g., ARIMA models, could account for linear trend by differencing the observations and then modeling the resulting differences as, ideally, a stationary time series.
As this approach gets further away from the classic linear modeling approaches more familiar with policy evaluators, we instead follow the linear modeling approaches found in the ITS and CITS literature.
%, borrowing from the idea of autoregression to improve model plausibility and consequent inference.
% and focus primarily on improving the linear regression approach
For this alternative direction, however, see, e.g., \citet{stoffer2006time}

In this paper we first lay out the ITS problem and its classic treatment.
We then, in Section~\ref{sec:simulation}, describe the simulation procedure that allows for a simple autoregressive structure, illustrating with an example taken from the Mecklenberg County evaluation.
We then provide our two primary extensions mentioned above---seasonality and smoothing---in Sections~\ref{sec:seasonality} and \ref{sec:smoothing}.
We offer some general cautions and concluding remarks at the end.
Our supplementary offers some further commentary and some extensions.
In particular, we give further justification of the modeling choices we suggest, discuss how to adjust for time varying covariates, and finally gives a brief overview of the accompanying publicly available R package that implements everything discussed.
We do not focus on the causal inference reasoning behind ITS: our work focuses on assessing whether there was change; the question as to why requires additional work, and for that we refer the reader to sources such as \cite{cook2002experimental}.
Causal interpretation of an ITS finding can be fragile in some contexts; see, e.g., \cite{baicker2019testing}.

\section{Notation and Setup}
We have a single treated unit. 
We observe this unit at several time points before treatment (e.g. a policy change) as well as for several time points after.
For example, consider Figure~\ref{fig:data}, showing two time series: the proportion of all arrests in Mecklenberg for each month for a period before and after a reform effort \citep{Redcross:te}, and the total number of warrant arrests each month before and after a major reform effort in New Jersey \citep{Golub:I9dpB0lF}.

Based on the trend of the unit before the policy change, we will extrapolate to determine what we would see post policy had business continued as usual.
For example, if we have observed a steady but slow increase in our outcome, we would project that steady but slow increase into the post policy period.
If what we actually observe deviates from that projected trend, we know that something has changed our system to cause this departure.
The core assumption behind an Interrupted Time Series design is stability; everything rests on the assumption that, absent any impact, our unit would evolve as it has been.

\begin{figure}
\centering
\begin{subfigure}[t]{.48\textwidth}
  %\centering
  \includegraphics[width=\textwidth]{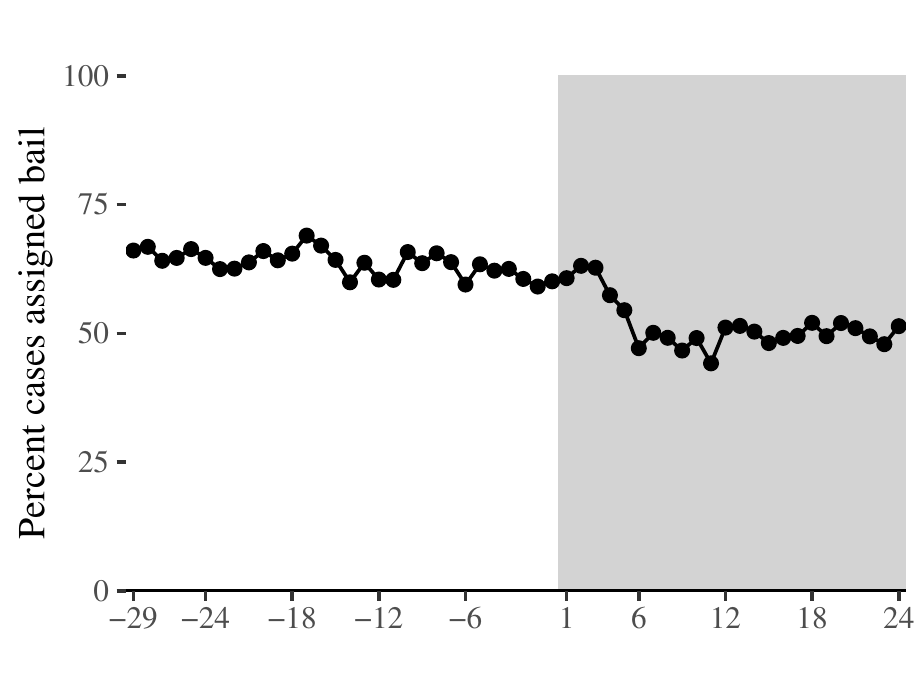}
  \caption{Monthly proportion of all arrests that assigned bail (or detention) in Mecklenberg.}
  \label{fig:meckdata}
\end{subfigure}
\begin{subfigure}[t]{.48\textwidth}
  %\centering
  \includegraphics[width=\textwidth]{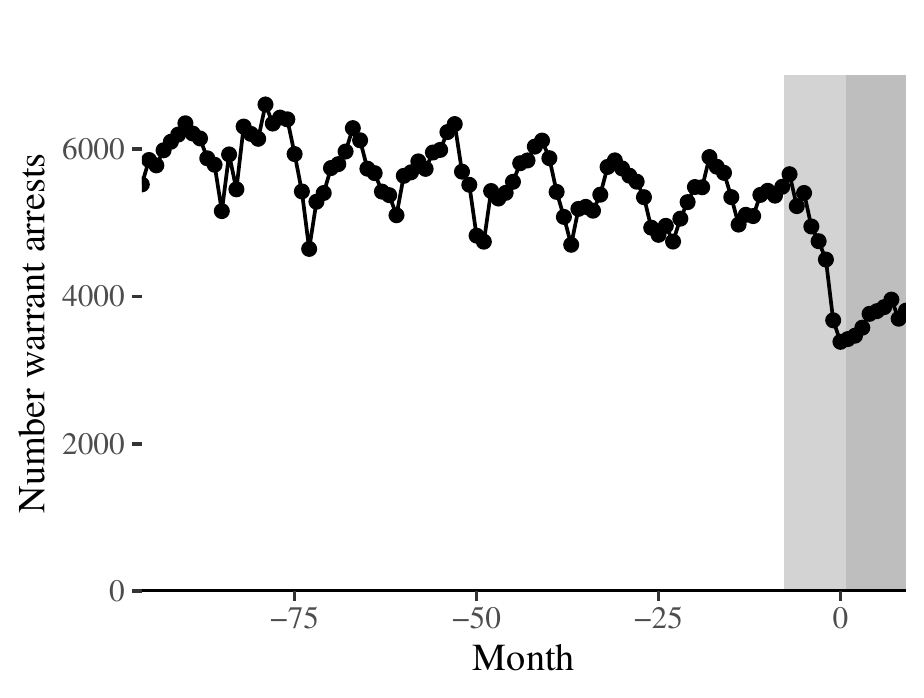}
  \caption{Monthly number of warrant arrests made in New Jersey.}
  \label{fig:njdata}
\end{subfigure}
\caption{Two example Interrupted Time Series.}{The dark grey indicates the post-policy era. $t_0 = 0$ in these figures, with pre-policy time being non-positive. Right side shows evident seasonality.  Left side suggests some auto-correlation which may be due to seasonality or other unknown factors.}
\label{fig:data}
\end{figure}

We borrow from the potential outcomes viewpoint (for overview, see \citet{CausalInferenceText} or \citet{Rosenbaum:2012ul}) to make the above more precise.
We have a single unit, and we can either treat it (invoke policy change) at time $t_0$, or not treat it at all.
Let $Y_t(0)$, $t = t_{min}, \ldots, t_{max}$, be the sequence of outcomes we would observe if we did not ever treat our unit.\footnote{Note the subscript here does not denote the unit, as is typically seen, but rather the time of observation for our single unit.}
Let the corresponding $Y_t(1)$ be the outcomes we would observe if we did treat the unit at $t_0$.
We could allow $Y_t(1) \neq Y_t(0)$ for $t \leq t_0$ if we allowed anticipatory effects of treatment, 	i.e., if the unit knows it will be treated it may change before the time of treatment.
In this work, we make the further assumption, however, that there is \emph{no anticipation of treatment}, i.e., that $Y_t(1) = Y_t(0)$ for all $t \leq t_0$.
In some cases, to achieve this assumption, one can move the point of treatment onset earlier, e.g., to when a policy was initially being planned rather than its official adoption date.

The impact of policy at a specified time $t$ is then $\Delta_t \equiv Y_t(1) - Y_t(0)$.
Our observed data consist of a single treated unit, so the $Y_t(1)$ are observed for all $t > t_0$.
If we had the ability to estimate $Y_t(0)$ we could immediately estimate $\Delta_t$.
This converts our estimation problem to a missing data problem \citep{Rubin:2005hf}.
Within this framework, uncertainty around the difference is entirely dependent on uncertainty in our estimation of $Y_t(0)$.

ITS analysis estimates the $Y_t(0)$ by fitting a trend (i.e., model) to the pre-policy data and extrapolating to post-policy timepoints.
%Linear models are the most common, but other models are also possible.
We next discuss how this estimation is typically conducted, and identify some problems with it.
We then offer an augmented modeling approach with corresponding inference procedures.

\paragraph{Remark.}
We analyze at the group level by aggregating individual data within each month.
We might imagine instead analyzing at the individual level, but this will bring in further complexity from, e.g., individuals being in multiple months (e.g., from multiple arrests in our context), and unknown correlation structure of individuals within a given month; aggregation avoids this.
Furthermore, migration of individuals into and out of the policy region could further exacerbate the difficulties with individual trend approaches.
The aggregation avoids these problems by focusing on the ``health'' of the policy unit rather than the impact on individuals.
Results are then regarding changes at the larger unit level, which can impact interpretation.
That being said, without strong individual level predictors, aggregation will surprisingly not have a high cost in power; the variation in the month-to-month averages is a reflection of individual variation (as well as shared month shocks) and so the fewer data points is coupled with less residual noise for those points.
See \citet{angrist2008mostly}, Chapter 3, for a more detailed discussion of this general rule.
For further discussion on aggregation see \cite{Bloom:2005vc}, Appendix D.
For some dangers with aggregation if the number of units being aggregated changes significantly, see \citet{Ferman2015}.

\subsection{Classic ITS analysis}

In a classic ITS analysis one would fit the simple linear regression model of
\begin{equation}
Y_t = \beta_0 + \beta_1 t + \sum_{k=t_0+1}^{t_{max}} \Delta_k \ind{t=k} + \epsilon_t \label{eq:classicITS}
\end{equation}
with $\epsilon_t \stackrel{\text{iid}}{\sim} N(0,\sigma^2)$ and the $\ind{t=k}$ 0/1 indicators of whether $t=k$ for each post-policy time point $k$.
This model will perfectly fit all post-policy months, meaning the estimates of $\beta_0$ and $\beta_1$ will only depend on pre-policy months.
The $\widehat{\Delta}_k$ are then the specific impact estimates for each month $k$, capturing the departure of $Y_t$ from the projected $\hat{\beta}_0 + \hat{\beta}_{1}t$.
Under a homoskedasticity assumption, we can obtain standard errors and conduct inference for the estimated $\Delta_k$, because we assume the variation post-policy is the same as pre-policy.
These standard errors will be driven by, and be no smaller than, $\hat{\sigma}$, the estimated residual standard deviation (see Appendix A for derivation).

Nearly equivalent to the above, one can simply fit the model to the pre-policy data only, dropping the post-policy dummy variables:
\begin{equation}
Y_t(0) = \beta_0 + \beta_1 t + \epsilon_t . \label{eq:control_equation}	
\end{equation} 
We then, for any point $t > t_0$ in the post-policy era, predict via extrapolation,
\begin{equation*}
 \hat{Y}_{t}(0) = \hat{\beta_0} + \hat{\beta_1} t ,
\end{equation*}
which results in an impact estimate at month $t$ of
\[ \widehat{\Delta}_{t} = Y_{t}^{obs} - \hat{Y}_{t}(0) .\]
These point estimates will be identical to the $\widehat{\Delta}_t$ from Model~\ref{eq:classicITS}.
However, Model~\ref{eq:control_equation} makes the connection to the potential outcomes framework most clear: our model predicts, via extrapolation, $Y_t(0)$ for all $t > t_0$.
We fit our model to pre-policy data, data unaffected by the policy (by assumption), and then use our fitted model to impute (predict) the missing $Y_t(0)$ for $t > t_0$.\footnote{By contrast, instead of not using post policy data at all in the fitting process, some will instead put a structure on the post-policy impact as well, such as with
\[ Y_t = \beta_0 + \beta_1 t + \delta_0 \ind{t > t_0} + \delta_1 \ind{ t > t_0 } (t - t_0 - 1) + \epsilon_t , \]
with $\ind{ t > t_0 }$ being a 0/1 indicator of $t$ being after $t_0$, the end of the pre-policy era.
Now the parameters $\delta_0$ and $\delta_1$ form a model of effects for the impact (in this case the impact begins at size $\delta_0$ and grows by $\delta_1$ each month, and $\Delta_t = \delta_0 + \delta_1(t-t_0 - 1)$ for $t > t_0$.  This allows the post-policy data to inform the estimated residual variance.}
%We advise against such models as the post-policy time points are now informing the estimated residual structure of the pre-policy trend.
%In particular, misspecification of the functional form could im distort other parameter estimates in unexpected ways.}

These models produce valid inference under the modeling assumptions, in particular the strong assumption of the linear trend continuing into the post-policy period.  
As a model check, the linear trend can be assessed in the pre-policy period; if there are strong deviations  pre-policy, then extrapolation should be done with skepticism.
The \emph{causal} interpretation, however, relies on any found deviation being only explainable by the policy change; it is a substantive question whether there were other factors or changes that happened concurrently or after the policy reform, producing changes in outcomes that should not be ascribed to the policy.

One concern with these approaches is that there may be effects that operate in windows of time causing adjacent months to have similar outcomes beyond the underlying model.
For example, the pattern of month-to-month averages in Mecklenberg (see Figure~\ref{fig:meckdata}) could contain local correlations of months around what is a generally linear trend (we discuss the case of cyclic seasonal trends such as shown in Figure~\ref{fig:njdata} in Section~\ref{sec:seasonality} below).
%On the right hand side we see what appears to be a strong seasonality effect for total arrests in New Jersey, but even beyond those effects it appears as if there may be autocorrelation as well.
%On the right hand side we see what appears to be a strong seasonality effect for total arrests, but even beyond those effects it appears as if there may be autocorrelation as well.

If we do not model temporal dependence, we are assuming that, other than the underlying linear trend, there is no dependence between months beyond the explicit model.
For example, if month $t$ were surprisingly high, this would not imply any other month, such as month $t+1$, would have any particular value.
To produce more principled inference we therefore extend Model~\ref{eq:control_equation} to allow for neighboring residuals to be correlated.
This better captures how the time series can ``wander'' from the linear trend.%\footnote{Depending on the frequency of the observations, this correlation may be increasingly more critical to model. For example, if we had daily measures, the autocorrelation could be quite strong. Yearly measures may be less of a concern, depending on context.}

A simple approach is to model local dependence using an ``AR1'' model that uses the residual in the prior time period as a predictor of the residual of the next.
For example, we can specify the residual of Model~\ref{eq:control_equation} to be
\begin{equation}
 \epsilon_t = \rho \epsilon_{t-1} + \omega_t \mbox{ with } \omega_t \stackrel{\text{iid}}{\sim} N( 0, \sigma^2 ) .\label{eq:residual}	
\end{equation}
The parameter $\rho$ governs how much autocorrelation we have.
If $\rho = 0$ the residuals are in fact independent.
Higher values of $\rho$ means deviations from trend tend to be similar, month-to-month.
A $\rho > 1$ would mean a successive observation would be some percent larger than the last, in expectation, and thus the series would exponentially move away from the trend line; we therefore require $\rho < 1$.

An easy way of fitting such a model is to fit the lagged \emph{outcome} model of
\begin{equation}
Y_t = \tilde{\beta}_0 + \tilde{\beta}_1 t + \tilde{\beta}_2 Y_{t-1} + \tilde{\epsilon}_t  \mbox{ with } \tilde{\epsilon}_t \stackrel{\text{iid}}{\sim} N(0,\tilde{\sigma}^2) \label{eq:ITS_lagged_outcome}
\end{equation}
to the pre-policy time points $t = 2, \ldots, t_0$.
The initial month has to be dropped as it has no lagged month.
Up to how the parameters are interpreted, this model is equivalent to the lagged residual model. 
In particular, as the derivations in Appendix A show, we have $\rho = \tilde{\beta}_2$, $\beta_1 = \tilde{\beta}_1 / (1-\tilde{\beta}_2)$ and $\beta_0 = \tilde{\beta}_0 / (1-\rho) - \tilde{\beta}_1 \rho / (1-\rho)^2$.
The residuals in the lagged outcome model are, under our residual autoregressive model, again independent, corresponding to the $\omega_t$ from Model~\ref{eq:residual}.
See Appendix A for additional discussion.

Once this model is fit, we use it to extrapolate a reasonable counterfactual prediction of $Y_T(0)$ for any timepoint $T > t_0$ of interest.
In the next section, we discuss how to do this with simulation.

\section{Extrapolating pre-policy trends via simulation} \label{sec:simulation}

Impacts are estimated by extrapolating the pre-policy model to a post-policy timepoint, $T > t_0$, of interest.
It not obvious how to use the model to form counterfactual predictions when using autoregressive structure.
In particular, for $T > t_0 + 1$, if the treatment has impacted point $T-1$, we cannot use the observed $Y_{T-1}$ as our lagged covariate for our prediction because $Y_{T-1}$ is not an observed $Y_t(0)$, but rather a $Y_t(1)$; any treatment impact in our lagged covariate will contaminate our imputation of $Y_T(0)$.
Second, assessing uncertainty for a point $T$ dependent on prior points is, mathematically, not entirely transparent.
We therefore assess uncertainty and form predictions via simulation.\footnote{One could instead use maximum likelihood and asymptotic approximations given the defined residual structure; we argue the parametric simulation approach we use provides a flexible and easily extendible alternative.}
In the next subsection, we first consider the case where we are willing to assume the lagged model is correct and we knew with certainty the parameters $\tilde{\beta}_0, \tilde{\beta}_1, \tilde{\beta}_2$, and $\tilde{\sigma}^2$ of our lagged model.
This case is not quite valid since we do not know these parameter values and so our uncertainty is not fully captured; we include it for clarity of exposition.
We then, in the following subsection, extend to our actual proposed method that incorporates the additional uncertainty of these parameters.

\subsection{Extrapolating with known parameters}
We initially assume the model of Equation~\ref{eq:ITS_lagged_outcome} and that our parameters $\theta = (\tilde{\beta}_0, \tilde{\beta}_1, \tilde{\beta}_2, \tilde{\sigma}^2)$ of the pre-policy model are known.
We also have observed $Y_{t_0}$, the last point in the pre-policy era.
%, and, under our model, know $Y_
%\[ Y_{t_0} = \beta_0 + \beta_1 t_0 + \beta_2 Y_{t_0-1} + \epsilon_{t_0} .\]

Using this, we can simulate $Y_{t_0+1}$ by drawing a new $\epsilon_{t_0+1}^* \sim N(0,\tilde{\sigma}^2)$ and calculating
\[ Y^*_{t_0+1} = \tilde{\beta}_0 + \tilde{\beta}_1 (t_0+1) + \tilde{\beta}_2 Y_{t_0} + \epsilon_{t_0+1}^* .\]
This simulated outcome is a plausible post-policy outcome, given our model.  
We can then simulate an outcome for $t_0+2$ using $Y^*_{t_0+1}$, drawing a new $\epsilon_{t_0+2}^*$ and adding up the components just as for $t_0+1$.
Our second simulated outcome depends on our first.  If our first is elevated due to a positive residual, our second will also be elevated.  
We then simulated our third, using the second, and continue in this manner until we reach $T$, and are left with a prediction for $Y_{T}$.
By this point we have generated an entire sequence of plausible outcomes, given our model.
Furthermore, this simulation process has fully captured the autoregressive structure.

Our final prediction $Y_{T}$ is a noisy prediction: it could be high or low depending on the residual draws.
This noise is the key to capturing uncertainty.
Both to get a more precise prediction and also to model the prediction uncertainty, we repeat the simulation process many times, for each iteration beginning at $t_0$ and $Y_{t_0}$ and simulating a new time series. 
We then calculate the average of these series to get our final prediction:
\[ \widehat{Y}_{T}(0) = \frac{1}{R} \sum_{r = 1}^R Y^{*(r)}_{T},  \]
where $R$ is the total number of simulated series and $r$ indexes these simulated series.

For inference, the middle 95\% of our simulated $Y_{T}^{*(r)}$ forms a 95\% prediction interval of what we would expect to see, had the pre-policy trend continued.
If what we actually see, $Y_{T}^{obs} = Y_T(1)$, lies outside of this interval, we have evidence our model does not extrapolate to time $T$, suggesting that something happened to change our model.
This would be evidence of an impact of either the policy change or some other event within the system.

We can subtract the prediction interval from the observed $Y_{T}$ to obtain a prediction interval for the deviation from the predicted trend (this is the quantity that could potentially be viewed as an impact).
%This is essentially a prediction interval for how much the series has deviated from its predicted trend.
This prediction interval correctly captures the month-to-month variability of the observed trend; see Appendix A.

The major caveat to this process is we do not know the true $\theta$; we instead have an estimate $\widehat{\theta}$. 
If we simply plug in $\widehat{\theta}$ our inference will be overly optimistic as we have not taken uncertainty in the estimation of the parameters themselves into account; we do that next.

\subsection{Incorporating uncertainty in the parameters}
To capture parameter uncertainty we use a method rooted in Bayesian thinking and taken from \cite{gelman2006data}.
It also has ties to the parametric bootstrap \citep[see, e.g.][]{davison1997bootstrap}.
The idea is this: instead of using $\hat{\theta}$, draw a random vector of parameters $\theta^*$ for our model given our observed pre-policy data.
This randomly drawn vector of parameters is itself a plausibly true value, just as we were drawing plausibly true values for the $Y_t$, above.
We then simulate a sequence of $Y_t^*$ using the simulation process described above but with the $\theta^*$ (and still starting at $Y_{t_0}$) to get a plausibly true prediction conditional on the parameters.
This two-step process captures the uncertainty in model estimation as well as uncertainty in extrapolation due to the autoregressive structure and residual error.
The distribution of the $Y_T^*$ over repeated iterations gives an overall predictive distribution that is integrated over both these components.

To get our distribution of plausible $\theta^*$, we use the (estimated) standard errors from the original model fitting process.
In particular, we draw a random $\beta^* = (\beta_0^*, \beta_1^*, \beta_2^*)$ vector from a multivariate normal centered at $\hat{\beta} = (\hat{\beta}_0, \hat{\beta}_1, \hat{\beta}_2)$ with a variance-covariance matrix based on the estimated variance-covariance matrix from the linear model fitting procedure (the $\sigma^{2*}$ term is handled similarly).
This is implemented using the \verb|sim()| function in the R package \verb|arm|.
The \verb|arm| package was written specifically for this form of uncertainty quantification, and is the companion package to \cite{gelman2006data}.

This approach is essentially Bayesian: the parameter draw step is similar to drawing a plausible value from a posterior distribution on the true $\theta$ (the implied prior here is implicitly a flat prior on the coefficients, roughly meaning that we are not differently preferring one value of $\theta$ over another).  
Under this view, the simulations constitute a posterior predictive distribution for $Y_{T}$ and the $\hat{Y}_{T}$ is the posterior mean predicted outcome given all the pre-policy data and the model (see \citet{gelman1996posterior} for a discussion of posterior predictive distributions).
Further, under this view, the final prediction interval can be interpreted as a posterior predictive interval for $Y_{T}(0)$.
Imputing missing potential outcomes in this way follows the approaches discussed in, e.g., \cite{Rubin:2005hf}.
Regardless, the core feature of this approach is that we end up with a range of plausible values for $Y_{T}(0)$ that incorporate the natural variation in the data as well as uncertainty about the parameters of our model.

The validity of the range of plausible values depends on the model being correctly specified.
We believe this approach to uncertainty quantification renders model dependency more transparent (salient) than a classic maximum likelihood analysis or regression approaches.
For example, we here see more explicitly the importance of the correct specification of the initial linear trend and the homoskedasticity assumption.
We are not making more or different assumptions than the classic approaches with autoregressive specifications, but rather are making the identical assumptions more explicit.
We do avoid some of the asymptotic approximations used in maximum likelihood inference.

\subsection{Case Study: Mecklenberg County and the proportion of cases assigned bail}

Mecklenberg instituted a series of reforms including changing their pre-trial risk assessment tool to a tool called the Public Safety Assessment (the PSA).
These reforms were designed to reduce the negative impacts on arrestees while maintaining public safety; the goal is to identify and release those defendants unlikely to not appear at future court hearings or break further laws while awaiting trial, while imposing monitoring on the remainder.
One outcome of interest in evaluating the effectiveness of this program is the rate of bail setting (what proportion of cases resulted in the assignment of bail or outright detention) as compared to outright release.
See \citet{Redcross:te} for further discussion.

To investigate this we fit Equation~\ref{eq:ITS_lagged_outcome} to the Mecklenberg data displayed on Figure~\ref{fig:meckdata}.
Our estimated coefficients are $\widehat{\beta}_0 = 45$, $\widehat{\beta}_1 = -0.12$, and $\widehat{\beta}_2 = 0.26$.
The lagged outcome term ($\hat{\beta}_2$) is not significantly different from 0.
We see that the pre-policy trend does appear roughly linear.
The lack of significance of our autoregression term suggests that there is little autocorrelation after the linear trend is accounted for, but keeping it in our simulation incorporates the additional uncertainty that even a small amount could bring.
Dropping the lagged term from our model would be imposing the assumption of independence, which, given substantive knowledge of seasonality effects on criminal and policing behavior, is not tenable.
The failure to find a significant correlation could be a power issue.

\begin{figure}
\centering
\begin{subfigure}[t]{.48\textwidth}
  %\centering
  \includegraphics[width=\textwidth]{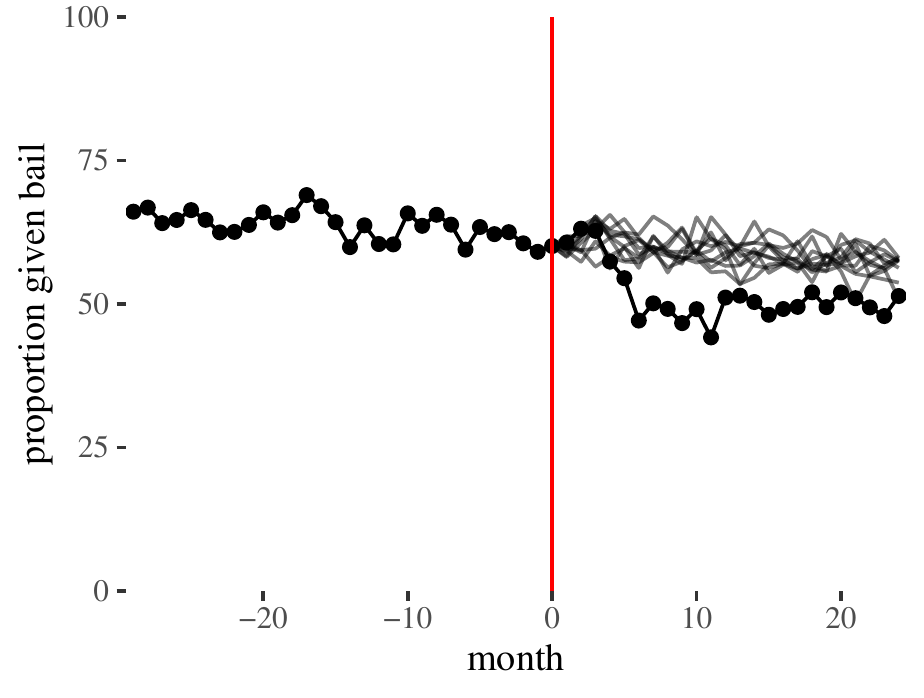}
  \caption{Ten extrapolated series.}
  \label{fig:meck10traj}
\end{subfigure}
\begin{subfigure}[t]{.48\textwidth}
  %\centering
  \includegraphics[width=\textwidth]{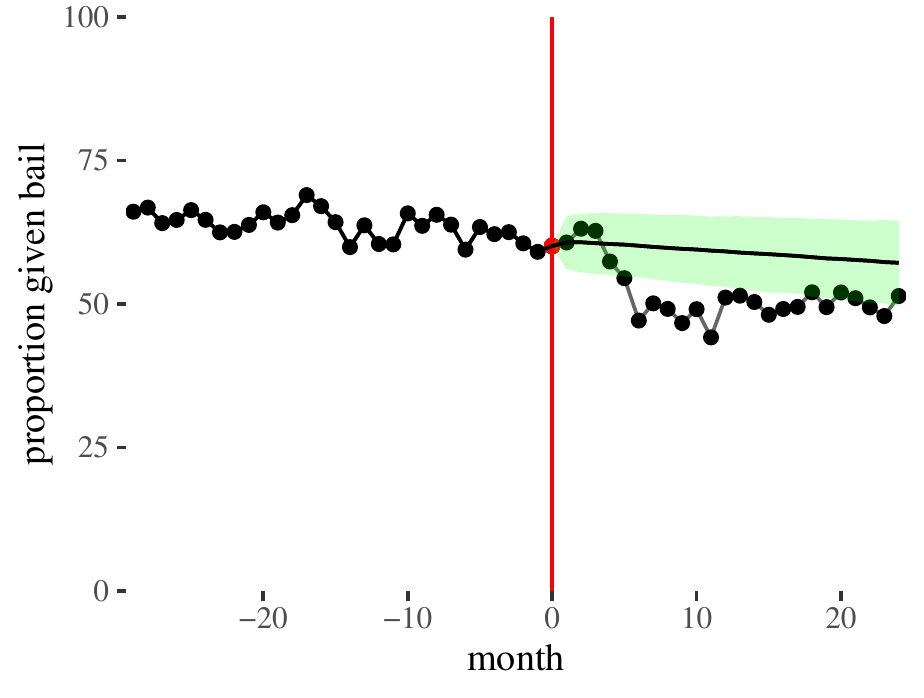}
  \caption{Envelope from 10,000 simulations.}
  \label{fig:meckenvelope}
\end{subfigure}
\caption{Results of Mecklenberg analysis.}{At left ten sample simulated series along with observed data. At right the overall envelope of plausible series given pre-policy data. We see that for many post-policy months the proportion of cases assigned bail is  not in the range of likely bail rates, suggesting that there was a more rapid decline of bail-setting after the policy change than expected given the slow decline of the pre-policy trend.}
\label{fig:meckresults}
\end{figure}

Using our model we can generate trajectories starting at $t_0 = 0, Y_0 = 60.1$.
Ten such extrapolations are on Figure~\ref{fig:meck10traj}.
We generate 10,000 such extrapolations based on 10,000 draws of possible parameters $\theta$, and summarize by, for each time point, taking the middle 95\% range of values.
We plot these as an envelope on Figure~\ref{fig:meckenvelope}.

Overall we see evidence of a reduction of the use of bail.
Pre-policy trends do not tend to fall as far as what actually occurred.
We also see that the observed outcomes for the first four months after the policy change are still potentially following the pre-policy trend; the departure is only really significant at month 5 and 6.
At this point, actual bail mostly levels off at the reduced rate of around 50\%. 
Patterns such as these raise important issues of how to ascribe the change: was this drop at month 5 due to the policy shift, or due to some subsequent intervention that may or may not have been part of the policy?
In this case, there is some qualitative evidence that Mecklenberg continued to reinforce their policy change with additional trainings of court agents, which could have caused this delayed impact.

The nominal impact is the difference of the \emph{projected} trend and the actual, which means the change in the overall level of an outcome does not necessarily mean there is a measured impact.
In this case, for example, we see the overall linear pre-policy trend projecting a steadily decline of bail assignment.
This means that at around 2 years post policy we cannot rule out an absence of impact: those bail levels may have been reached regardless, considering the pre-policy declining trend, but at a later time than with the policy change.

But then again, the further out an extrapolation the greater our dependance on the model being correctly specified, both statistically and as a representation of a dynamic and complex system.
The statistical model can extrapolate assuming the general model fit to pre-policy, but the assumption that these trends would continue indefinitely becomes substantively less plausible the further away from the transition we go. 
The greater uncertainty in later months is only due to estimation error, and is dependent on the assumption that the pre-policy process would have continued unabated in the absence of the policy change.
In particular, we cannot know if alternate measures would have been taken had the policy not been imposed or if the system would have naturally reached some change point given the dynamics.

Overall, there are three sources of uncertainty to attend to in such analyses, with only the first two  quantifiable: (1) parameter estimation error for the model, (2) the natural variation due to month-to-month changes and associated auto-regression, and (3) model specification.
%It is important to think about and track these three aspects of variation when evaluating ITS data.

\section{Seasonality effects}
\label{sec:seasonality}

In New Jersey, when a person is arrested the arresting officer can (1) serve a summons, where the officer gives the arrestee a court date for a future appearance and then sends them home, or (2) serve a warrant, which could result in detention until the resolution of the case.
One consequence of a policy rooted in risk assessment might be to change policing behavior towards only giving warrants for the more serious offenses.
An outcome of interest that assesses this is the total number of warrant arrests made.

Counts can be more difficult to model than proportions.
Figure~\ref{fig:njdata} shows a strong periodic trend across the years, with reduced number of arrests when it is winter, and more in summer.
%As shown on Figure~\ref{fig:njdata}, the total count of arrests can have a strong seasonal trend, 
In fact, average temperature in a month (a good proxy for season) is found to predict total arrests quite strongly; see Figure~\ref{fig:temperature}.
These seasonal cycles are likely due to factors such as increased time spent indoors and away from the public eye during the colder winter months.

\begin{figure}
\centering
\begin{subfigure}[t]{.48\textwidth}
  \centering
  \includegraphics[width=\textwidth]{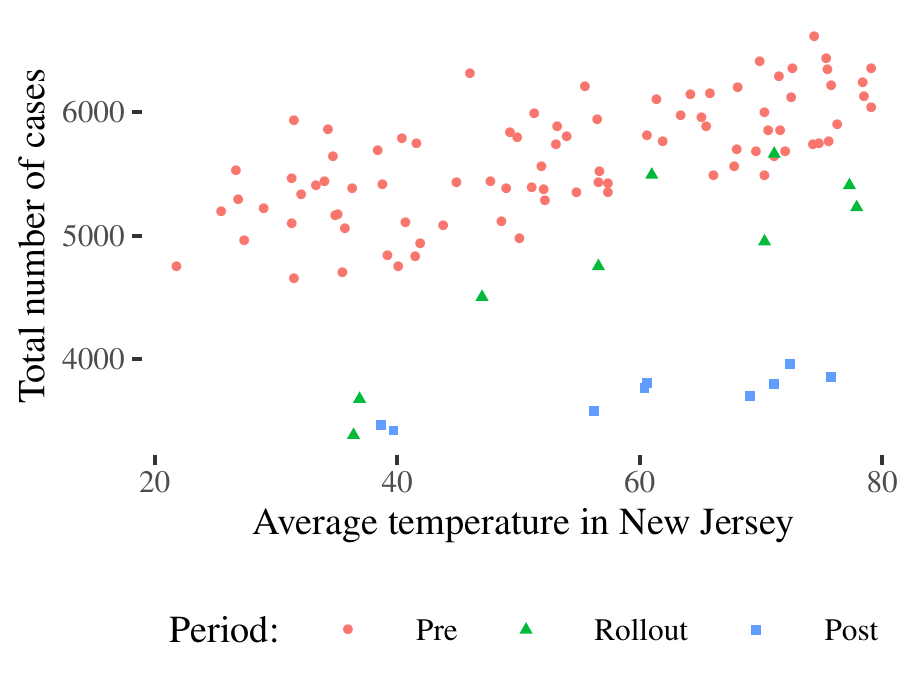}
  \caption{Relationship of average temperature to number of warrant arrests.}
  \label{fig:temperature}
\end{subfigure}
\begin{subfigure}[t]{.48\textwidth}
  \centering
  \includegraphics[width=\textwidth]{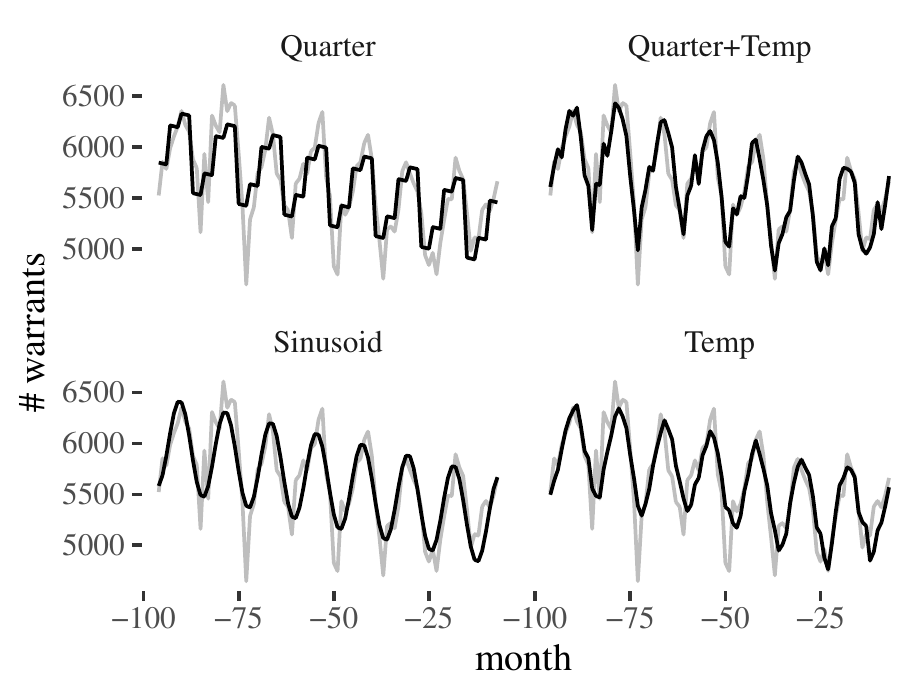}
  \caption{Four seasonality models for number of warrant arrests}
  \label{fig:all_season_models}
\end{subfigure}
\caption{Seasonality in the total arrests data in New Jersey.}
\label{fig:seasonality_plots}
\end{figure}

Fitting a simple autocorrelation model would miss the cyclic nature of our trend, which means we have clear model misspecification and which, in this case, results in substantial loss of power (as we show below).
We instead extend our linear model to model the periodic trend.
The autoregressive element would then allow local departures from the overall seasonality model, just as we had local departures from the linear model above.

There are several ways one might capture a periodic seasonality structure with linear regression.
A simple approach is to include dummy variables for the four seasons.
The following model, for example, has the first quarter as a baseline, has three offsets for the other three quarters, and also allows an overall linear trend:
\[ Y_t = \beta_0 + \beta_1 t +  \gamma_2 Q_{2t} + \gamma_3 Q_{3t} + \gamma_4 Q_{4t} + \epsilon_t , \]
with $Q_{2t}, Q_{3t},$ and $Q_{4t}$ $0/1$ indicators for being in the 2nd, 3rd, and 4th quarters of the year.
A second approach is to use a covariate that is predictive of outcome and is itself periodic, such as, in our case, monthly average temperature in the region:
\begin{equation*}
Y_t = \beta_0 + \beta_1 t + \beta_2 Temp_t + \epsilon_t, 
\end{equation*} 
where $Temp_t$ is a measure of average temperature for month $t$.
The periodic nature of our data is then driven by the periodic nature of our time-varying covariate.
These general approaches can easily be combined:
\begin{equation}
	Y_t = \beta_0 + \beta_1 t + \beta_2 Temp_t + \beta_3 Q_{2t} + \beta_4 Q_{3t} + \beta_5 Q_{4t} + \epsilon_t .	\label{eq:temp_model} 
\end{equation} 
One potential concern with seasonal dummy variables is the resulting curve will be a step function rather than a smooth curve, with steps at pre-specified points that are not data driven.
We could alternatively fit a sinusoidal trend by building two covariates that correspond to the sin and cos of the month (rescaled to have a yearly period).
Linear combinations of these two covariates allow for sinusoidal curves that can be smoothly shifted left or right.
For example:
\[ Y_t = \beta_0 + \beta_1 t + \rho_1 sin( 2 \pi t / 12 ) + \rho_2 cos( 2 \pi t / 12 ) + \epsilon_t . \]
Different coefficient values for $\rho_1$ and $\rho_2$ control where the peaks and valleys of this trend are.

To illustrate these four fitting approaches, see Figure~\ref{fig:all_season_models}, which shows simple fits (without lagged variables) to the pre-policy data.
Of the four models, the model with both quarter and temperature has the best pre-policy fit, with an estimated residual standard deviation of 192 compared to 250 and above for the other models.

\subsection{Seasonality with autoregressive residuals}

%We next demonstrate how we might investigate the impact of policy on a series of warrant arrests with our synthetic data.

Once a seasonality model is selected, we again are faced with how to fit the autoregressive residual structure in a simple way that also lends itself to simulation.
We cannot simply include the lagged outcome, as this lagged outcome includes the lagged periodic structure.
We therefore include the lagged values of the covariates used to model seasonality along with the outcome; this subtracts out the lagged structural component of the trend, resulting in a corrected model that puts the autoregression solely on the residuals.
See Appendix B for a derivation of this result, along with some alternative estimation strategies.

For example, for Model~\ref{eq:temp_model} we would have $X_t = (1, t, Q_{t2}, Q_{t3}, Q_{t4}, Temp_t)$.
We would then fit the following regression:
\begin{align*}
 Y_t &= X_t'\beta - X_{t-1}' \beta_\ell + \rho Y_{t-1} + \omega_t ,
\end{align*}
with $\beta$ our primary trend and $\beta_\ell$ our lagged ``anti-trend'' (generally $\beta \approx \beta_\ell$, with an exact equality if we fully believe our lagged model).
There is a small technical caveat: the lagged covariates can frequently be collinear with the contemporaneous covariates, thus producing an overall design matrix that is not full rank.
For example, if we include a linear time component by including the covariate $X_{t,2} = t$ as one of the columns of our design matrix, the design matrix with our lagged covariate of $X_{t,k} = t - 1$ will clearly be fully collinear with $X_{t,2}$.\footnote{This colinearity is why the simple lagged linear trend model does not have an extra term beyond the lagged outcome itself.}
This can also happen with periodic covariates such as $X_{t,k} = \sin( a t )$.
This colinearity is easily resolved: simply drop collinear columns (in particular the intercept and time variables), allowing the remaining parameters to estimate the combined influence of both the primary observation and the structural component of the lagged outcome.

\subsection{Case study: New Jersey and the number of warrant arrests}

We next analyze the data on warrant arrests shown on Figure~\ref{fig:njdata} with our seasonality model.
We fit Model~\ref{eq:temp_model} with the autoregressive residual model of $\epsilon_t = \rho \epsilon_{t-1} + \omega_t$.
We set $t_0 = -8$ due to evidence that there was some preparatory restructuring and changes made in advance of the policy's official start date to ensure a smooth launch; by setting $t_0=-8$ we increase the plausibility of our no anticipation assumption.
%As before, we cannot directly fit this model with ordinary least squares due to the unobserved residuals.
We, following the above, extend our model to include the lagged outcome and lagged covariates, giving
\begin{align*}
  Y_t &= \beta_0 + \beta_1 t + \beta_2 Temp_t + \beta_3 Q_{2,t} + \beta_4 Q_{3,t} + \beta_5 Q_{4,t} + \\
  &\qquad \beta_6 Temp_{t-1} + \beta_7 Q_{2,t-1} + \beta_8 Q_{3,t-1} + \beta_9 Q_{4,t-1} + \rho Y_{t-1} + \omega_t	.
\end{align*}
We then generate the predictive envelope on Figure~\ref{fig:nj_results_unsmooth} by following the process described above.

\begin{figure}
\centering
\begin{subfigure}[t]{.48\textwidth}
  %\centering
  \includegraphics[width=\textwidth]{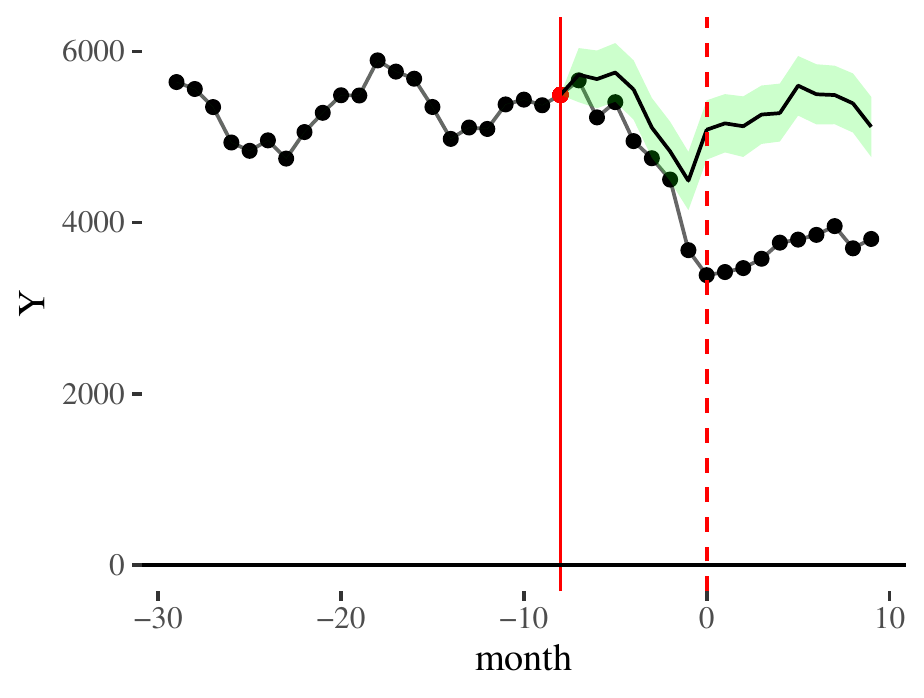}
  \caption{Projection with seasonality.}
  \label{fig:nj_results_unsmooth}
\end{subfigure}
\begin{subfigure}[t]{.48\textwidth}
  %\centering
  \includegraphics[width=\textwidth]{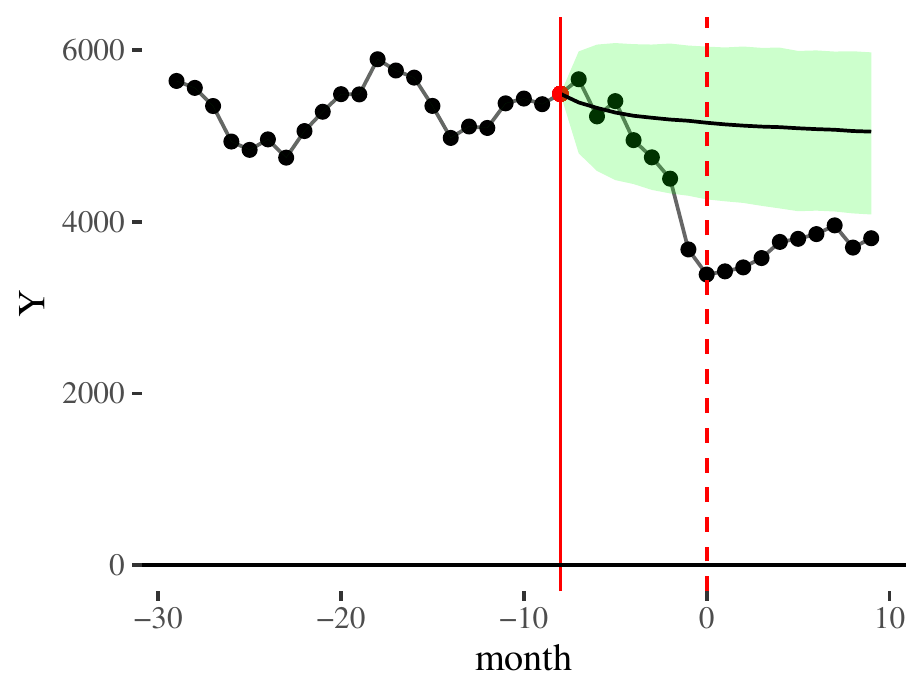}
  \caption{Projection without seasonality}
  \label{fig:nj_no_season}
\end{subfigure}
\caption{Prediction envelopes for number of warrant arrests in New Jersey.}{Time period ($x$-axis) truncated to show more detail of model fitting in post-policy era. (a) shows seasonality model, (b) shows model with no seasonality. See raw data on Figure~\ref{fig:njdata}.}
\label{fig:nj_results}
\end{figure}

By comparison, if we had not included a seasonality model and instead simply fit our simple linear trend model, we get Figure~\ref{fig:nj_no_season}.
The model without seasonality has more autocorrelation (estimated as 0.77 vs. 0.69), because points near each other are correlated due to the periodic trend around the base linear model, while the seasonality model captures and removes these dependencies.
This autocorrelation allows for large deviations from trend in the simulated extrapolated series, and thus we see a large confidence envelope.
In general, without the seasonality model we are not able to take advantage of the seasonal structure of the data, but the autoregressive element does capture that there is local dependence, resulting in a conservative inference.

One might ask whether Mecklenberg should also be fit with a seasonality component.
Generally, to reliably estimate a seasonality structure, we would need several cycles of the seasons; Mecklenberg is ``too short'' to ascertain that structure.
In this case, we rely on the estimated autocorrelation to capture the overall uncertainty.
Determining when to fit the more complex model vs. not is an important area for future work, but we found that with eight years of data, and a clear seasonal trend, the seasonality model was easily estimable.
For Mecklenberg, however, seasonality models were quite unstable.

\section{Inference and smoothing} \label{sec:smoothing}

Reading the envelope graphs from the above analyses can be somewhat confusing as there are multiple post-policy months with some of them having observed outcomes lying outside of the predictive envelope and others not.
In this section, we discuss inference more formally and discuss how to increase power by averaging the outcomes of post-policy months together.
For this averaging we can either average a fixed range of months, or use methods akin to a sliding window by nonparametrically smoothing the observed trends to account for month-to-month variation.
This sliding window approach is appealing in that we can display an entire curve of impacts post-policy, which allows for a more nuanced interpretation of how a policy may have evolved over time.

\subsection{Inference}
Consider the null hypothesis of there being no change in the pre-policy trend (and that we have correct model specification).
In this case, our simulated series are all plausible forecasting series, given the pre-policy data.
For any given point $T > t_0$, we can therefore examine the distribution of simulated values at $T$ to see how much variability we would see under the null hypothesis.

In this view, we use our observed outcome $Y^{obs}_T$ as the observed value of a test statistic: we compare this observed value to the simulated values that capture what our model says is possible.
If the observed value is outside the central range of these simulated values (which we consider our reference distribution), we reject the null that the pre-policy trend continued unabaited (again assuming the pre-policy model is correct).
We could do this for each $T > t_0$.

While reasonable and sound, there are two concerns: first, we have a multiple testing issue.  If the series is long enough, we are bound to find some points outside their respective predictive ranges simply due to random fluctuation. 
Second, we have a power issue.  
We are comparing our test statistic, a potentially highly variable single point $Y^{obs}_T$, to a distribution of simulated values $Y^*_{T}$ that all themselves could be quite variable.
If the policy caused a modest reduction in $Y^{obs}_t$ for all $t > t_0$, it is possible that no individual $T > t_0$ would look significantly reduced when examined in isolation.

As a contrast to testing a specific point in time, we might instead test for a systematic and sustained shift in the outcomes over a range of times post-policy.
%when looking across all our post-policy time points, or all time points in an interval, we can ask whether the entire line is consistent with the null.
%This is a different question than testing a specific point $T$.
In order to test a larger sequence of time points, we need to combine our observed data into some sort of average and compare that \emph{average} to the distribution of averages we would have likely seen under the null.

The simplest approach to do this is to simply average all the outcomes in a pre-specified range of months post-policy.
We then compare this simple average to the distribution of simple averages calculated from the distribution of plausible trajectories.
The key point is once we have our distribution of plausible trajectories, we can test our null hypothesis by comparing a summary statistic of our outcome to the distribution of that same summary statistic calculated on our trajectories.
To be specific, take our observed series $Y = (Y_1, \ldots, Y_T)$ and calculate our summary $t^{obs} = t( Y )$, where $t(\cdot)$ a function that takes our data and summarizes it in some way (e.g., by calculating the average of $Y_{t_0+1}, \ldots, Y_{T}$).
Next, for each simulated series $Y^{*(r)}$, calculate $t^{*(r)} = t( Y^{*(r)} )$, and then calculate the $\alpha/2$ and $1-\alpha/2$ quantiles $t_{(\alpha/2)}$ and $t_{(1-\alpha/2)}$ of these $t^{*(r)}$.
Our prediction interval of what value of the summary statistic we would expect to see is then $CI = (t_{(\alpha/2)}, t_{(1-\alpha/2)})$.
If $t^{obs} \not \in CI$, we reject our null hypothesis.
We calculate nominal $p$-values using the percentile $q$ of our observed $t^{obs}$, with $p = \min( q, 1-q )$ (for a two-sided test).

Testing in this way is akin to posterior predictive checks of model fit \citep{rubin1984bayesianly,guttman1967use}: we want to know if the model fit to pre-policy data fits our post-policy observed data.
If it does not, we reject the model, i.e., conclude that something changed our trajectory.
%One can calculate nominal $p$-values by calculating the proportion of time the observed test statistic is exceeded by the reference distribution.
Our $p$-values are called \emph{posterior predictive $p$-values}, and do not necessarily have strictly valid frequentist properties, but they are argued to generally be conservative \citep{meng1994posterior}.
Also see \cite{robins2000asymptotic}.

\subsection{Smoothing}
In investigating a place-based initiative we generally want to understand the evolution of the impact over time.
For example, with Mecklenberg, it appears as if the policy induced a large reduction in the rate of bail setting a few months into the post-policy period, with that level of bail setting generally sustained over time.
If we only use the simple averaging method from above, and did not look at the overall graph of impacts, we would lose this nuance.
But the raw graph is noisy, making trends somewhat difficult to discern.
We therefore might want to smooth the trend in the graph to, as much as possible, remove month-to-month variation.
Smoothing is when one locally summarizes a trend to remove some variation, ideally without imposing a global structural so local structure is preserved \citep{cleveland1993visualizing}.
Smoothing is generally nonparametric, and can be done with splines, averaging within a sliding window, or using loess (Locally Weighted Smoothing) (\citet{cleveland1981lowess}, but see \citet{cleveland1993visualizing}).
Smoothing can make communication with various stakeholders easier, as it removes random variation that may draw ones attention if not removed; see, e.g., discussions in \cite{starling2019targeted}.

We can easily use smoothing coupled with our inferential approach above.
In particular, we smooth each simulated time series using a specific (pre-specified) method.
We then compare the distribution of these smoothed time series to the actual time series smoothed in exactly the same way.
Under our null hypothesis, the smoothed observed trajectory should be exchangeable with any of the smoothed simulated trajectories. 
Our smoothed estimate at a given timepoint $T$ is now our test statistic, and the distribution of smoothed estimates of our simulated series our prediction distribution of what values we might have expected.
This should have greater power: we are now examining the overall trend in the neighborhood of $T$, potentially increasing precision as idiosyncratic monthly variation gets averaged out.

One caveat is that if we smooth across $t_0$ we can cause the smoothed line of our observed series to artificially deviate \emph{pre-policy} since the post-policy points will be included in the local average near the policy change.
Similarly, the pre-policy timepoints near $t_0$ can drag the smoothed post-policy timepoints near $t_0$ towards their values, potentially masking impacts.
To avoid this, one can smooth the post-policy series only, not including any pre-policy points; if this is done, then it needs to be done for both the simulated series as well as the observed series.
The key is to implement the same process on all series, simulated and observed, to maintain the validity of the comparison.

\begin{figure}
\centering
\begin{subfigure}[t]{.48\textwidth}
  %\centering
  \includegraphics[width=\textwidth]{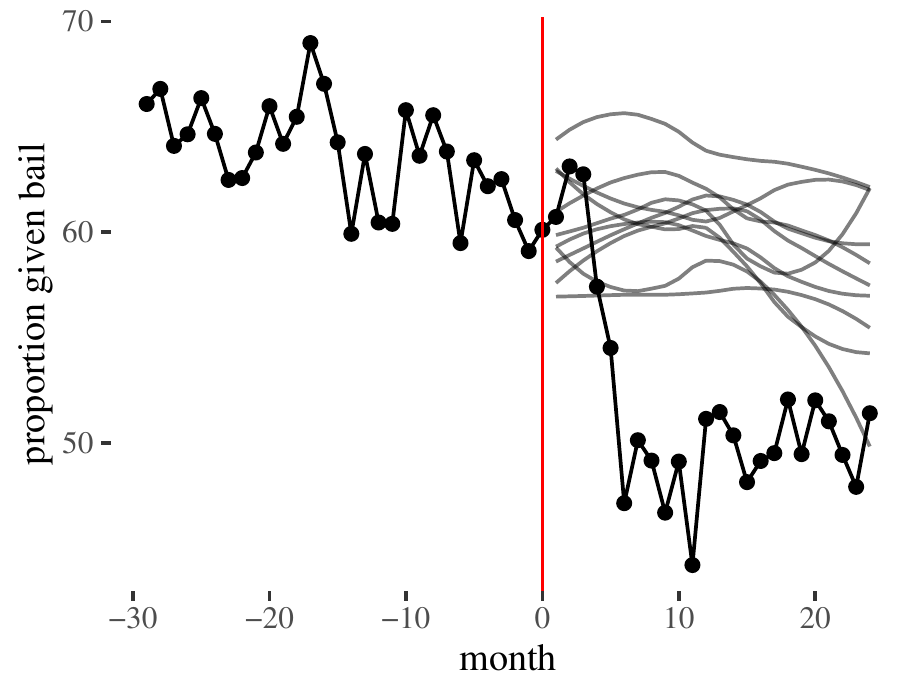}
  \caption{Ten extrapolated smoothed series}
  \label{fig:meck10trajsmooth}
\end{subfigure}
\begin{subfigure}[t]{.48\textwidth}
  %\centering
  \includegraphics[width=\textwidth]{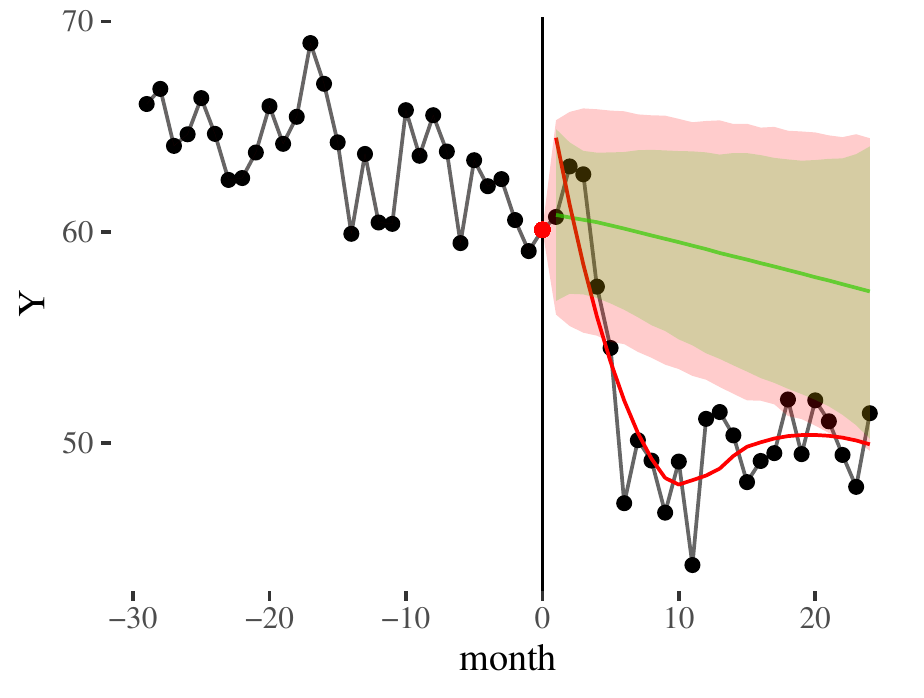}
  \caption{10,000 smoothed simulations}
  \label{fig:meckenvelopesmooth}
\end{subfigure}
\caption{Results of Mecklenberg analysis (with smoothing).}{Left shows how smoothed trajectories have less variability than the raw series did. Right compares smoothed envelope with envelope without smoothing. We see less variability. The red line denotes the smoothed observed trend to be compared to the envelope and counterfactual predicted trend.}
\label{fig:meckresults_smooth}
\end{figure}

\subsection{Mecklenberg County, continued}

We continue our Mecklenberg example by showing how to improve power using both averaging and loess smoothing.
We initially average the outcomes for the initial 18 months after $t_0$.
In our data, we observe an average bail rate of 52\%.  
The middle 95\% prediction interval of the averages of our simulated series ranges from 55\% to 64\%.
We therefore conclude that something changed the pre-policy trajectory so we are seeing lower average rates of bail-setting than we would have expected.
If we take the difference to get estimated impacts we obtain a 95\% confidence interval (technically a credible interval) for the true average impact being in $(-3\%, -12\%)$.
To get a point estimate for the average impact, we average the simulated averages, predicting an average bail setting of $59\%$ and estimated impact of $-7$ percentage points.

If we look at a tighter range of months (which we would ideally have pre-specified) of 6 months to 18 months, we observe an average of 49\%, a corresponding prediction interval of 54\% to 64\%, and an estimated impact of between $-5$pp and $-15$pp.
Choice of summary measure can substantially matter here as they will differently weight what are often quite heterogeneous impacts across time.

We also use loess smoothing to smooth the post-treatment trajectory.
We first smooth our observed series with a loess smoother fit to the post-policy data only to avoid any influence of pre-policy points on our resulting line.
We then fit the same smoother to each of our simulated series, ignoring the pre-policy points there as well.
Results are on Figure~\ref{fig:meckresults_smooth}.
Figure~\ref{fig:meck10trajsmooth} shows 10 smoothed trajectories in the post-policy period.
Figure~\ref{fig:meckenvelopesmooth} shows the envelope based on these trajectories, along with the smoothed observed line and, in the background, the original envelope without smoothing.
The smoothed observed curve is arguably easier to read than the raw data.
We also see precision gains from the smoothing process, which stabilizes the estimation.
Also note the wider envelope at far left; this is due to loess smoothers being more variable at endpoints.

Smoothing does require specifying a tuning parameter of how much to smooth.
For loess, for example, we essentially specify what fraction of the data should be used to calculate the smoothed outcome at each time point.
If we smooth a lot, then local variation in the structure will be removed, but the lines will be more stable.
If we smooth little, then we do not really average local points, and thus our variance will remain high.
This is a bias-variance tradeoff in the estimation and visualization.

\subsection{Smoothing with seasonality.}
When the model has a seasonality component causing oscillation, a simple loess smoother might dampen the oscillations, creating a smoothed series that is more flat than the data.  
This not only looks odd, but can be deceptive.
But, as discussed above, we can smooth in any fashion we choose, as long as we smooth our observed data in the same way as the simulated.
This allows for the following multi-step smoothing approach that smooths the residual variation around the structural component of a seasonality model.
For each time series (observed or simulated) smooth as follows:
\begin{enumerate}
\item Fit a working seasonality model to the data.  This is not the original seasonality model, but a new model. There is no need for lagged outcomes or uncertainty estimation in this model. As before, we can choose to fit to post-policy data only, all data, or pre-policy data.
\item Predict all the outcomes given our seasonality model.
\item Calculate the residuals by subtracting the predicted outcomes from the actual outcomes.
\item Fit a loess smoother (or some other smoother) to the residuals (again choosing whether to focus on post-policy only or all data).
\item Add the smoothed residuals back to the predictions to get a final smoothed curve.
\end{enumerate}
This process strips the estimated approximation of the structural component from the series and sets it aside to prevent it from being smoothed or averaged out.
Step (5) puts it back so our final series still maintains the overall structure.
In particular, any estimated seasonality component will not get smoothed out.
The key idea is that our model used for smoothing does not need to be a correctly specified model; it is purely to set aside any seasonal structure so it does not get over-smoothed.

% and so if there is a strong cycle it will not be removed by the smoothing process.

%If we repeat this for each simulated series, we get a null distribution of what distribution of smoothed series we would see if the model were correct.
%We then do a final step of fitting our observed data with the same smoothing model, and compare this series to the distribution of predicted series.

\subsection{New Jersey, continued}

To demonstrate smoothing with a seasonality model we extend our analysis of warrant arrests.
We compare two methods of smoothing, using the same model for extrapolation, but calculating our residuals using two different models.
In one we use the base (non-lagged) model with quarter and temperature, and in the other we fit the sinusoidal model without temperature.
Our second model intentionally smooths away month-to-month variability due to fluctuating temperature in both our simulated and observed series, even though we use the temperature to fit and extrapolate our data to obtain our predictive series before smoothing.
The results are on Figure~\ref{fig:nj_results_smoothed}. The left has preserved the month-to-month variation predicted by the temperature changes, giving a more jagged sequence.
The right, by contrast, is smoother, showing underlying structure more clearly. 
%We see that there are increased numbers of warrant arrests, and the difference appears fairly constant over time.
%This structure is a summarization of the larger gross trends, however, and we should acknowledge that there could be a bias-variance tradeoff at play.

\begin{figure}
\centering
\begin{subfigure}[t]{.48\textwidth}
  \includegraphics[width=\textwidth]{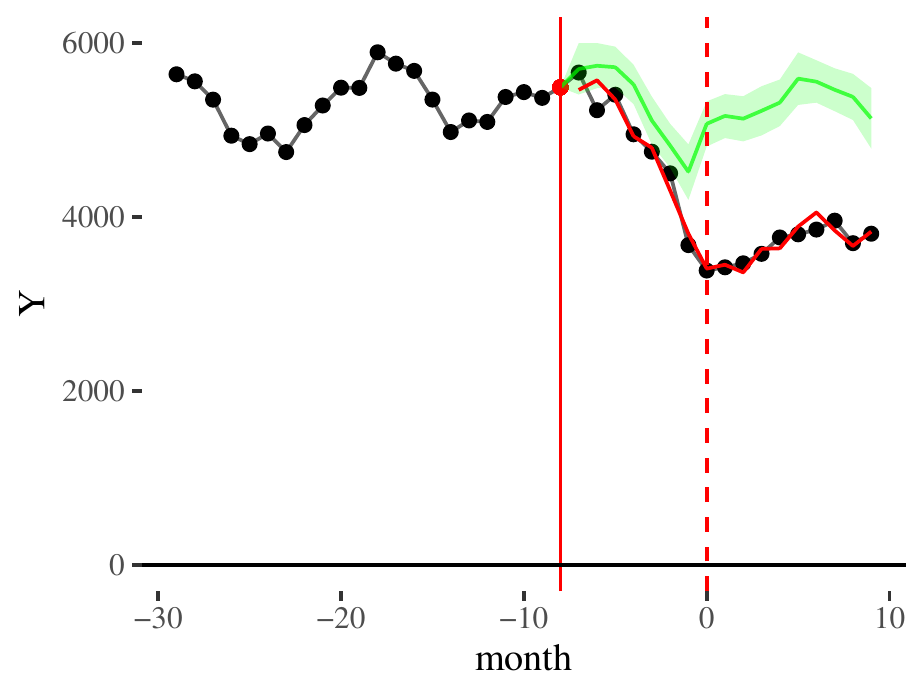}
  \caption{Base smoothing model}
  \label{fig:nj_results_smoothed_base}
\end{subfigure}
\begin{subfigure}[t]{.48\textwidth}
  \includegraphics[width=\textwidth]{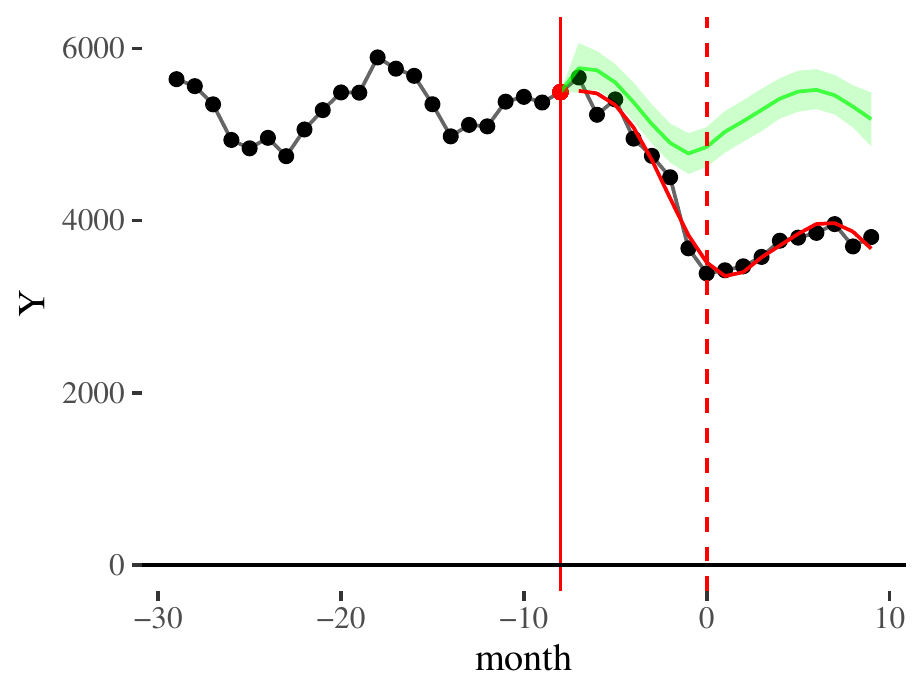}
  \caption{Sinusoidal smoothing model}
  \label{fig:nj_results_smoothed_sin}
\end{subfigure}
\caption{Prediction envelopes for number of warrant arrests using smoothing.}{Time period and $y$-axes truncated to show more detail of model fitting in post-policy era.}
\label{fig:nj_results_smoothed}
\end{figure}

\paragraph{Remark.} 
If we do not fit the same model to the same range of data across all series, smoothing in the observed series could cause different distortions than in the reference series.
This could create systematic differences even if the null of no treatment impact were true.
However, if there is a large initial treatment impact, a model fit to the full observed series could be misspecified.
This could give an odd smoothed series for the observed data.
Regardless, as long as the model fitting process is held to be the same, then comparing the observed series to the reference distribution is still valid for testing.
We recommend selecting a smoother that is not overly dependent on the pre-policy patterns, but instead naturally fits to the observed post-policy data.
In particular, we suggest fitting the seasonality model to the post-policy data only.

\section{Conclusion}

We have demonstrated a simple modeling (linear regression with lagged outcomes and covariates) and simulation framework for capturing uncertainty for Interrupted Time Series designs.
These designs often appear when attempting to assess the impact of a policy change on a single region of interest when there are no good comparison regions available.

Our modeling framework allows for the easy incorporation of seasonality models and of smoothing in a straightforward manner.
It also naturally allows for incorporation of autoregressive structure to better account for overall uncertainty.
Finally, we argue that this approach naturally lends itself to generating clear visualization of impacts and transparent reporting of results.

In this work, we have examined ITS designs with many pre-policy timepoints; with fewer timepoints estimating the autoregressive component of the model will generally be much more difficult.
In maximum likelihood approaches, it is known that this can cause bias and poor coverage when there are only 5 or so observations (see, e.g., \cite{st2016validity}).
We leave whether simulation, simulation which specifically incorporates the uncertainty in the estimated lagged coefficients, would help in these short ITS designs to future work.

%Our framework, while useful, also has some areas that call for development.
%In particular, especially due to explicitly taking into account autocorrelation, this approach does not lend itself to straightforward power equations.
This approach could also be extended to power calculations.
Minimal detectable effect size (MDES) and power depend on several factors: the number of cases per month, the month-to-month variability beyond natural variation due to the cases, the number of months of pre-policy data, and the desired window of predicting impacts after the policy implementation. 
Each of these can heavily influence the ability to detect effects.
One way forward is to again turn to simulation.
In particular, given specific parameterized values for the factors listed above, one could repeatedly simulate a dataset, and then analyze that dataset using the above simulation approach as an inner step.
For each initially simulated dataset we would then record the width of the simulated extrapolations.
The average width of these prediction intervals at each time point could then be tied to MDES.

Finally, the modeling itself could also potentially be extended and enriched to better capture some data contexts.
For example, if the number of individual cases changed substantially over the course of a series, we might want to let our residual error be a function of sample size to capture differing levels of precision \citep[see, e.g.,][]{Ferman2015}.
One approach would be to regress residual size onto number of cases, giving an intercept and slope which would represent core month-to-month variability and within-month variability, and use this decomposed variation in the autoregressive model.

With ITS, there are some concerns with interpretation, in particular in the case of a dynamic system.
For example, if the impacts in early post-policy months are creating a feedback loop (e.g., changing patterns in detention causing changes in the patterns of new charges) then the mix of individual cases constituting the overall region may be changing as a result of the policy change.
This further underscores that interpreting impacts has to occur at the region level, which naturally takes these changes into account.
In particular, a reduction of bail rates could potentially be due to the policy changing the cases themselves, rather than be due to changes in how cases are being handled.
Ideally we thus should focus on measures that are of interest when viewed at the aggregate level.

And finally, fundamentally, we note that all that this type of analysis can show us, using this method or any other, is that the trend has changed in a surprising way. 
Why it did so, the statistics cannot answer.
The researcher in the end must turn toward substance matter knowledge and argument to defend the proposition that a found change was caused by the policy shift.

\bibliographystyle{chicago}
\bibliography{references2}

%\title{Supplement for\\Simulating for Uncertainty \\with Interrupted Time Series Designs}
%\author{Luke Miratrix \\Harvard Graduate School of Education}
%\date{\today}

%\maketitle
%\thispagestyle{empty}
%\pagenumbering{gobble}
%\singlespacing

%\begin{abstract}

%\end{abstract}

\newpage

\begin{center}
\huge
Supplementary Material for \\
``Using Simulation to Analyze \\ Interrupted Time Series Designs''
\end{center}

This online supplement contains additional discussion and results to complement the main paper. Appendix A contains a few lemmas and an extended discussion of the technical details. Appendix B gives further explanation of why we need to include lagged seasonal covariates, and also discusses some alternative modeling approaches one might take.  Appendix C gives an overview of how we could adjust for time-varying covariates (e.g., changing mix of cases) in an analysis.  Appendix D has a demonstration of the R package that accompanies this paper.

\section*{Appendix A: A few lemmas and extended discussion}

We here provide the small lemmas and derivations that complement the text.  After we have a few further remarks about the approach, including some notes as to what to do if the autoregressive coefficient is estimated to be larger than 1 or less than 0.

\subsection*{Standard errors for classic OLS-based ITS}

\begin{lemma}
Under the classic OLS approach, the standard errors for the $\widehat{\Delta}_k$, $k > t_0$ are 
\[ \widehat{SE}\left[ \widehat{\Delta}_k \right] = \left( 1 + S_{00} + 2 k S_{10} + k^2 S_{11} \right)^{1/2} \hat{\sigma} \geq 1, \]
with the $S_{00}$, $S_{11}$ and $S_{01}$ being the elements of the variance-covariance matrix one would obtain for the coefficients of a simple intercept-slope regression of the outcome on the prepolicy data only.
\end{lemma}
\begin{proof}
The design matrix $X$ of our regression consists of a column of 1s, a column with the time values $1, \ldots, T$, and one column for each post-policy timepoint, where the column for time point $k$ has all 0s except a single 1 at row $k$.
Given this, $X'X$ is a $(2+K) \times (2+K)$ block matrix:
\[ X'X = \left[ {\begin{array}{cccccc}
   T & \sum_t t & 1 & 1 & \cdots & 1 \\
    & \sum_t t^2 & t_0+1 & t_0+2 & \cdots & t_0+K \\
    & 			& 	1 & 0 & \cdots & 0 \\
    & 			& 	0 & 1 & \cdots & 0 \\
   \cdot & \cdot & \cdot & \cdot & \cdots & \cdot \\
   &	&			&	 & \cdots & 1
  \end{array} } \right] = \left[ {\begin{array}{cc}
   A & B \\
   B' & I \end{array} }  \right], \]
where $K$ is the number of post-policy timepoints (and the number of our different $\Delta_k$ we are estimating).
Note the bottom-right block is a $K \times K$ identity matrix.
$A$ corresponds to what we would get from the simple linear regression of $Y$ on the months $1, \ldots, T$ with an intercept.

Classic OLS gives our standard errors for our coefficients as the diagonals of $(X'X)^{-1}\hat{\sigma}$.
We next calculate $(X'X)^{-1}$.
The inverse $(X'X)^{-1}$ will give another symmetric matrix defined by the block matrix
\[ (X'X)^{-1} = \left[ {\begin{array}{cc}
   (A-BI_K^{-1}B')^{-1} & -A^{-1}B(I_K-B'A^{-1}B)^{-1} \\
    & (I_K - B'A^{-1}B)^{-1} \end{array} }  \right] . \]
The standard errors for our impact estimates are governed by the bottom-right corner of the above.
We can simplify the bottom-right corner by using the Woodbury identity of 
 \[ (a + cbc')^{-1} = a^{-1} - a^{-1}c(b^{-1}+c'a^{-1}c)^{-1}c'a^{-1} \]
and $(-M)^{-1} = -M^{-1}$ to get
 \begin{align*}
 	(I_K - B'A^{-1}B)^{-1} &= I_K^{-1} - I_K^{-1}B'\left[ -A + B(I_K)B'\right]^{-1} BI_K^{-1} \\
 	&= I_K - B'\left[ -A + BB'\right]^{-1}B \\
 	&= I_K + B'\left[ A - BB' \right]^{-1}B .
 \end{align*}
We can then simplify the second term further.
Let $X_{pre}$ be the entries in the first $T-k$ rows and the first two columns of the design matrix.
This is the design matrix of the simple regression on pre-policy units only.
Then we have $A - BB' = X_{pre}'X_{pre}$.
To see this note that $BB'$ has the form
\[ BB' = \left[ {\begin{array}{cc}
   k & \sum_{j=1}^k T - k + 1 \\
    & \sum_{j=1}^k (T - k + 1)^2  \end{array} }  \right] . \]
Thus, subtracting $BB'$ from $A$ simply takes off the last elements of the sums.
Therefore $(A - BB')^{-1} = (X_{pre}'X_{pre})^{-1} = S$, the variance-covariance matrix for the coefficients of our simple pre-policy regression.
This gives
\[ (X'X)^{-1} = \left[ {\begin{array}{cc}
   (X_{pre}'X_{pre})^{-1} & -A^{-1}B(I_K-B'A^{-1}B)^{-1} \\
    & I_K + B' (X_{pre}'X_{pre})^{-1} B \end{array} }  \right] . \]

Finally, take the $r^{th}$ column of $B$, which corresponds to the $r^{th}$ post-policy period at time $t_0 + r$. This column is $(1, t_0 + r)$.
We then have
\begin{equation}
\left( B'(A - BB')^{-1}B \right)_{rr} = \left( B'SB \right)_{rr} = S_{00} + 2 (t_0 + r) S_{10} + (t_0+r)^2 S_{11} \geq 0 . \label{eq:delta_SE_expr}
\end{equation} 
The inequality is because $B'SB$ will have a nonnegative diagonal as $v'Sv \geq 0$ for any vector $v$ due to $S$ being positive semi-definite. 
Since this bottom right corner of the variance-covariance matrix is the above plus the identity matrix, we finally have our result.

\end{proof}

This lemma shows how our fundamental error is due to the initial 1 in the sum.  The remaining terms, in particular the $(t_0+r)^2 S_{11}$ term, correspond to the standard errors in the intercept and slope estimated on pre-policy data being extrapolated.
In particular, note how as $r$ increases, these terms grow quadratically in the variance and linearly in the standard error.
This is due to uncertainty in the slope causing increasingly large levels of extrapolated uncertainty.

\subsection*{Connection of lagged outcome model to residual dependence}

Our original goal was to allow for dependencies of the residuals in our linear model, but we used lagged outcomes instead.
These are the same thing, with different interpretation of the parameters in the model.
To see that a lagged model on $Y$ is the same as a model with lagged dependent residuals, take our simple regression model (Equation~\ref{eq:control_equation}) with its lagged residual model (Equation~\ref{eq:residual}).
Then
\begin{equation*}
\epsilon_{t-1} = Y_{t-1} - \beta_0 - \beta_1(t-1) 
\end{equation*} 
and
\begin{align*}
Y_t &= \beta_0 + \beta_1 t  + \left[ \rho \epsilon_{t-1} + \omega_t \right] \\
	&= \beta_0 + \beta_1 t  + \rho\left[ Y_{t-1} - \beta_0 - \beta_1(t-1) \right] + \omega_t \\
	&= \beta_0 + \rho( \beta_1 - \beta_0) + \beta_1(1-\rho) t + \rho Y_{t-1} + \omega_t \\
	&= \tilde{\beta}_0 + \tilde{\beta}_1 t + \tilde{\beta}_2 Y_{t-1} + \omega_t .
\end{align*}
This shows $\tilde{\beta}_0 = \beta_0 + \rho( \beta_1 - \beta_0)$, $\tilde{\beta}_1 = \beta_1(1-\rho)$, and $\tilde{\beta}_2 = \rho$.
Fitting our lagged outcome model will give estimates for $(\tilde{\beta}_0, \tilde{\beta}_1, \tilde{\beta}_2)$.
We can then convert them to our target $(\beta_0, \beta_1, \rho$); it is simply a different parameterization.
Some algebra gives the conversions to the residual model as
\begin{align*}
	\rho &= \tilde{\beta}_2 \\
	\beta_1 &= \frac{1}{1-\rho} \tilde{\beta}_1 \\
	\beta_0 &= \frac{1}{1-\rho} \tilde{\beta}_0 - \frac{\rho}{\left(1-\rho\right)^2} \tilde{\beta}_1 .
\end{align*}
The estimated residual variance is the variance of the $\omega_t$.

\paragraph{Properties of the residuals.} 
For our lagged residual model, we have $\EE{ \epsilon_t } = 0$ and, for any given $t$,
\begin{align*}
 var( \epsilon_t ) &= \EE{ (\epsilon_t - 0)^2 } \\
  &= \EE{ \rho^2 \epsilon_{t-1}^2 } + 2 \rho \EE{ \epsilon_{t-1} \omega_t } + \EE{ \omega_t^2 } \\
  &= \rho^2 \sigma_\epsilon^2 + 0 + \sigma^2 .
\end{align*}
This gives 
\[ var( \epsilon_t ) = \frac{1}{1-\rho^2} \sigma^2 .\]
We also have
\[ cov( \epsilon_t, \epsilon_{t-1} ) = \EE{ \epsilon_t \epsilon_{t-1} } = \EE{ (\rho \epsilon_{t-1}+ \omega_t) \epsilon_{t-1}  } = \rho \sigma_\epsilon^2 \]
meaning that the residuals have correlation $\rho$.

\paragraph{Prediction intervals capture month-to-month variability.}
To see why this is the case, note how our autoregressive series gives a conditional prediction for each time point: given time $T-1$, our prediction for time $T$ is the structural component plus the autoregressive part of the residual.
Under this view, write the final predicted outcome at time $T$ as an average of the conditional predictions given time $T-1$:
\[ \widehat{Y}_{T} = \frac{1}{R} \sum_{r = 1}^R \left[ \beta_0^{(r)} + \beta_1^{(r)} T  + \beta_2^{(r)} Y_{T-1}^{*(r)} \right] . \]
Even though there is no final $\epsilon_{T}^{(r)}$ in the above it is equivalent, up to simulation error, to the simple average of the simulated $Y_{T}^{*(r)}$ because the average of the $\epsilon_{T}^{*(r)}$ is 0.
This shows that the variation in our predictions combines the variation in the prediction itself with the additional $\epsilon_{T}$, which is the independent variability.
Under this decomposition, the variation due to the $\epsilon_{T}$ is the variation of the observed series, and the remainder is the variation of the structural trend and autoregressive variation.
All of this depends on correct model specification, in particular, the assumption that our observed post-policy series has the same month-to-month variation as our pre-policy series (i.e., homoskedasticity).
% We do not need to adjust for uncertainty in the observed $Y_{T}$ because the uncertainty in the observed $Y_{T}$ as compared to the true underlying trend is captured by the prediction interval being inflated by the final simulated residual $\epsilon^*_{T}$ being added to the prediction $\widehat{Y}_{T}$.

\subsection*{Handling overly large or small estimates of $\rho$}

Estimation involves uncertainty, and when fitting a lagged variable we have a range of possible coefficients for $\rho$ that could include values larger than 1 or less than 0.
This can cause difficulties; in particular, if the uncertainty on the coefficients carries the coefficient for the prior Y to more than 1, those associated projections will compound exponentially and be nonsensical.
This happens when there is little model stability in the fitting of the model (e.g., with only a few months of pre-policy data, in particular), or if there are large nonlinearities in the pre-policy.
In the latter case, the estimated $\rho$ coefficient can be estimated as relatively large to compensate for the model misspecification.
If the coefficient is negative, the predictions can oscillate, again in a nonsensical manner.
If either happens only in the extreme draws of the posterior, there is no major concern as the prediction intervals will trim these extremes.
If they are more frequent, the confidence intervals will give overly wide ranges that signal the model fitting issues.

\section*{Appendix B: Why include lagged seasonality covariates}

We start with a general structural model with autoregressive residuals of
\begin{align*}
 Y_t(0) &= f_\beta( X_t ) + \epsilon_t ,
\end{align*}
with $f_\beta(X_t)$ being a model of covariates $X_t$ (where $X_t$ is a vector of potentially time-varying covariates including $t$ itself) indexed by  some parameter vector $\beta$.
The $f_\beta(X_t)$ is the structural aspect of our model and our residuals are then $\epsilon_t = Y_t - f_\beta( X_t )$.  
We assume that once we remove the structure we have a stationary\footnote{``Stationary'' means the autoregressive structure is constant across time, i.e., that the auto-correlation remains the same.} autoregressive process given by Equation~\ref{eq:residual} on the residuals.
To connect to the above, we have heretofore assumed that $f_\beta(X_t) = \beta_0 + \beta_1 t$ with $\theta \equiv \left( \beta_0, \beta_1 \right)$.

With this more general model, using Equation~\ref{eq:residual} and the consequent $\epsilon_{t-1} = Y_{t-1} - f_\beta( X_{t-1} )$, we have
\begin{align}
Y_t &= f_\beta( X_t ) - \rho \epsilon_{t-1} + \omega_t \nonumber \\
&= f_\beta(X_t)  + \rho\left[ Y_{t-1} - f_\beta( X_{t-1} ) \right] + \omega_t \nonumber \\
&= f_\beta( X_t ) - \rho f_\beta( X_{t-1} ) + \rho Y_{t-1} + \omega_t . \label{eq:general_lag_model}
\end{align}

Now consider the case when $f_\beta$ is a linear model, with $f_\beta( X_t ) = X_t'\beta$.
This is the case presented in the main paper.
For example, for Model~\ref{eq:temp_model} we would have $X_t = (1, t, Q_{t2}, Q_{t3}, Q_{t4}, Temp_t)$.
Plugging $X'_t \beta$ in to the more general Equation~\ref{eq:general_lag_model} gives
\begin{align*}
 Y_t &= X_t'\beta - \rho X_{t-1}' \beta + \rho Y_{t-1} + \omega_t \\
	&= X_t'\beta - X_{t-1}' \beta_\ell + \rho Y_{t-1} + \omega_t ,
\end{align*}
with $\beta_\ell = - \beta \rho$.
The second line above shows that, if we do not insist on keeping the structure of the same $\beta$ in both the $X_{t}$ and the $X_{t-1}$ terms in the above, we can simply regress $Y_t$ on $X_t$, $X_{t-1}$ and $Y_{t-1}$, dropping the constraint of $\beta_\ell = -\beta \rho$.

We here repeat the technical caveat from the main text: the lagged covariates can frequently be collinear with the contemporaneous covariates, thus producing an overall design matrix that is not full rank.
For example, if we include a linear time component by including the covariate $X_{t,2} = t$ as one of the columns of our design matrix, the design matrix with our lagged covariate of $X_{t,k} = t - 1$ will clearly be fully collinear with $X_{t,2}$.
This colinearity is easily resolved, however: simply drop any collinear columns (in particular the intercept and time variables), allowing the parameters to estimate the combined influence of both the primary observation and the structural component of the lagged outcome due to that variable.

\paragraph{Remark.} This model fitting process relaxes some of the structure of the parameters.
In particular, we no longer enforce $\beta_\ell = -\beta \rho$.
However, given the consistency of estimation for linear regression, we immediately have that as $n$ increases our Ordinary Least Squares (OLS) approach will converge on the correct parameterization, giving overall consistency even if we drop this constraint.
Interestingly, as a model check, given the relaxation one could compare $- \hat{\beta} \hat{\rho}$ to $\hat{\beta}_\ell$.  
They should be the same, up to estimation error.

To further examine why we have the lagged covariates in our model, write the above as
\[ Y_t = (X_t - \rho X_{t-1})' \beta  + \rho Y_{t-1} + \omega_t . \]
This formulation suggests that our regression is, in effect, regressing our outcome onto the differences of our covariates (including any linear time term or intercept) with the lagged covariates scaled by our unknown $\rho$ parameter.
We have this differencing component because the lagged $Y_{t-1}$ also includes those structural components and we would not want them to be counted twice.
Overall, the lagged covariates allow us to estimate and then subtract out the lagged structural part of the $Y_{t-1}$, leaving the lagged residual as desired.

\paragraph{Alternate estimation strategies.}

We use the lagged outcome model and OLS due to its simplicity, and to take advantage of the ability to easily simulate from the fitted model.
Other options are possible, which we briefly acknowledge and outline next.

First, there is \emph{Feasible Generalized Least Squares (FGLS).} 
Here we would explicitly model the residual structure as AR1 and use feasible generalized least squares by first estimating the model using simple OLS, and then estimating the implied residual matrix with the empirical residuals.
While entirely viable, we wanted to take full advantage of the \verb|sim| function from the \verb|arm| package; we therefore instead use OLS in order to leverage existing packages to obtain the pseudo-posteriors.
That being said, FGLS coupled with classic maximum likelihood estimation would provide asymptotic confidence intervals based on the normal approximation.

One could also employ \emph{iterative model fitting.}
In particular, if we had the $\epsilon_{t-1}$ we could use them as covariates in our regression instead of the $Y_{t-1}$.
This motivates first estimating them using a model without lagged covariates, and then using those estimates in a second run of our model.
First fit $Y_t = f(X_t) + \epsilon_t$ using OLS.
Then calculate $\hat{D}_t = Y_t - \hat{Y}_t$ with $\hat{Y}_t = \hat{f}(X_t)$.
Then refit a model of 
\[ Y_t^{(2)} = f(X_t) + \rho \hat{D}_{t-1} + \omega_t \]
and calculate new differences
\[ \epsilon_t^{(2)} = Y_t - \hat{Y}_t^{(2)} = Y_t - \hat{f}^{(2)}(X_t) . \]
Repeat until convergence.
Under this approach to get a model fit, it is unclear how to simulate to get prediction uncertainty, or how to otherwise codify uncertainty of our model parameters correctly.
%We could just take our final estimated $\epsilon_{t}$ as fixed covariates and do the simple simulation above; how far off would that be?

Finally, one could simply specify the model and fit it using a Bayesian model-fitting package such as \verb|stan| \citep{carpenter2017stan}.
This would be functionally equivalent to the above, although there would no longer be need to incorporate lagged covariates or outcomes as the model could directly work with the latent residuals.
One would also have to explicitly choose priors, and this approach may feel less accessible to many practitioners than our approach rooted in classic linear regression.

\section*{Appendix C: Adjusting for time-varying covariates}

The monthly outcomes in our examples are summaries of individual data, and if the composition of the individual data are changing we might see changes in the aggregate summaries due to factors other than the policy reform being investigated.
For example, in our context, the charges associated with an arrest fall in different gross categories, and these categories tend to be treated differently due to having generally different levels of severity.
For example, in Mecklenberg, charges fall into traffic, misdemeanor, and felony and, as shown on Figure~\ref{fig:meck_changing_count}, the overall rates of misdemeanors is falling from before the policy shift to well after.
This means the mix of cases (as captured by being misdemeanor, felony, or traffic related) at each month is changing, and so post-policy we have proportionally more felony cases, which are typically more severe and would likely involve greater rates of detention.
Thus our outcome of interest, percent bail, may be being impacted by this changing case mix.
We therefore want a sensitivity check to test whether our overall findings are being partially or fully driven by this change in case composition, rather than change in how cases are being treated by the judicial system.
In other words, we want to adjust for case composition in our estimation process.

Post-stratification is one approach for handling this type of problem when one has access to the individual data composing the overall outcome.
With post-stratification, instead of calculating the simple average outcome for a month, we reweight the individual cases so that the proportion of each case type matches some canonical distribution.
For example, in Mecklenberg, 33\% of the charges are felony, 54\% misdemeanor, and 13\% traffic across all our post-policy months.
Therefore, for each month, we calculate an adjusted rate of bail by first reweighting the cases in that month to match these percentages before averaging to obtain the outcome. (This is the same as first calculating the proportion of cases given bail for each of the three categories, and then averaging these proportions with the weights of 33\%, 54\% and 13\%.) 
After doing this for each month, any differences in our adjusted outcomes across months can be ascribed to how the cases are being treated differently rather than being due to different types of cases, as captured by the covariates.

This attribution does rely on the assumption, however, that the measured covariates being adjusted for capture all the case differences that are both changing across time and are associated with the outcome.
For example, if the proportion of misdemeanors is going down, but the (unobserved) severity of the remaining misdemeanors is going up, then a change in outcome may still be connected to changing case mix as captured by this unobserved severity.
That being said, using post-stratification as an adjustment technique can help explore whether it is plausible that observed potential impacts are simply due to change in case mix, and provide a good way of conducting a sensitivity check on one's results.
This approach has ties to matching or propensity score reweighting.
In particular, we refer to template matching \citep{Silber:2014ge}, where the mixture of patients within each of a collection of hospitals are reweighted to match a reference distribution to ease comparability across the hospitals.
Here, we would match across time.

\subsection*{Reweighting for post-stratification}
The first step for reweighting is to identify the target case mix that we wish to reweight each month to match.
We recommend doing this via identifying a time period of interest (e.g., the year following the policy change) and calculating the proportion of cases of each type in that period, storing them as a vector of proportions $\pi^*$.
For Mecklenberg, to illustrate, this is 
\[ \pi^* = (\pi^{F*}, \pi^{M*}, \pi^{T*}) = ( 33\%, 54\%, 13\% ) .\]

The next step depends on the type of outcome being considered.
There are two core types of outcome: total counts and mean outcomes.  
The former is, for example, total number of cases where the defendant failed to appear (FTA) at a scheduled court date.
The second is, for example, the average days spent in jail for a case, or the proportion of cases that were given bail.
(The proportion of cases assigned bail can be thought of as the average of a 0/1 indicator variable for getting assigned bail.)

To adjust the outcomes based on the proportions of cases in each subgroup, we first calculate the aggregated outcome for each subgroup for each month $t$.  
Call these $Y_t^s$, with $s$ being in our case $F$, $M$, or $T$.
This could be total count, average outcome in the subgroup, or the proportion of the subgroup with a given outcome.

Next we re-weight the observed overall outcome (there is no modeling here) depending on type of outcome.
For means and proportions, we calculate
\[ Y^*_t = \sum_{s} \pi^{s*} Y_t^s .\]
For count, we calculate, if we let $N_t$ be total number of cases at month $t$ and $N_t^s$ total number of cases in subgroup $s$ at month $t$,
\[ Y^*_t = N_t \sum_{s} \pi^{s*} \frac{Y_t^s}{N^s_t} \]
This formula comes from first changing our outcome to the mean outcome in the group, getting the estimated mean for the whole month using the group weights, and then scaling back to raw count.

We use the above formula to adjust all months, both pre- and post-policy.  
This gives an adjusted time series where we have controlled for the strata considered.
This series could diverge from the raw series, if the proportions are changing and if the average outcomes $Y_t^s$ differ across $s$.

We now, at this point, simply fit our normal ITS model on the adjusted series.
This is effectively weighted regression, where we have re-weighted units at the individual level.
(It is a bit odd in appearance in that we do not model the individual level units but instead aggregate.)

\paragraph{Remark.}
Other adjustment approaches are also possible.
One alternate method of post-stratification would be to divide the cases into subgroups and fit several distinct ITS models, one for each subgroup.
Unfortunately, especially given our approach to handling seasonality and the autoregressive structure, we would want to account for possible dependence between the subgroups (e.g., they all may have higher or lower outcomes in a given month in some correlated fashion).
How to do this when each subgroup is fit separately is unclear.
This is why we instead reweigh each month by the subgroup proportions and then fit the ITS model to the resulting combined series.
%We find that this approach works fairly well, although there are some violations of assumed homoskedasticity that could arise if the proportions are changing a great deal.

One could alternatively implement a version of the above scheme with a single modeling step that includes covariates, such as the proportion of cases in the three categories, and trend by covariate interactions.
In particular, one would add interaction terms between the intercept and the time covariates and the proportion in each group to the model.  
The estimated coefficients would then be averaged together after the model fitting step.
This approach is most similar to ``controlling'' for variables in a regression.
These interaction terms effectively allow each subgroup to have its own structural model.
Unlike the simpler approach above, this approach could even be extended to allow for completely different time trends for the different subgroups and could be used to improve the face validity of the model if there were substantive reasons to believe such variation were present.

The simple version we presented above adjusts the raw data by aggregating the data with weights, and then conducts the analysis on a single aggregated dataset.
More complex versions are to in effect fit individual models to the subgroups, and then aggregate the model results.
This could in principle be more powerful as we are using the proportion of cases in a given month as a covariate to predict variation, which could lead to smaller standard errors.
Unless these proportions are highly predictive of outcome, the gains from this will likely be minimal.

All of these alternative approaches could potentially have benefits, especially if different subgroups were believed to have substantially different trends.
We leave exploring these directions more fully to future work.

\begin{figure}
\centering
\begin{subfigure}[t]{.48\textwidth}
  %\centering
  \includegraphics[width=\textwidth]{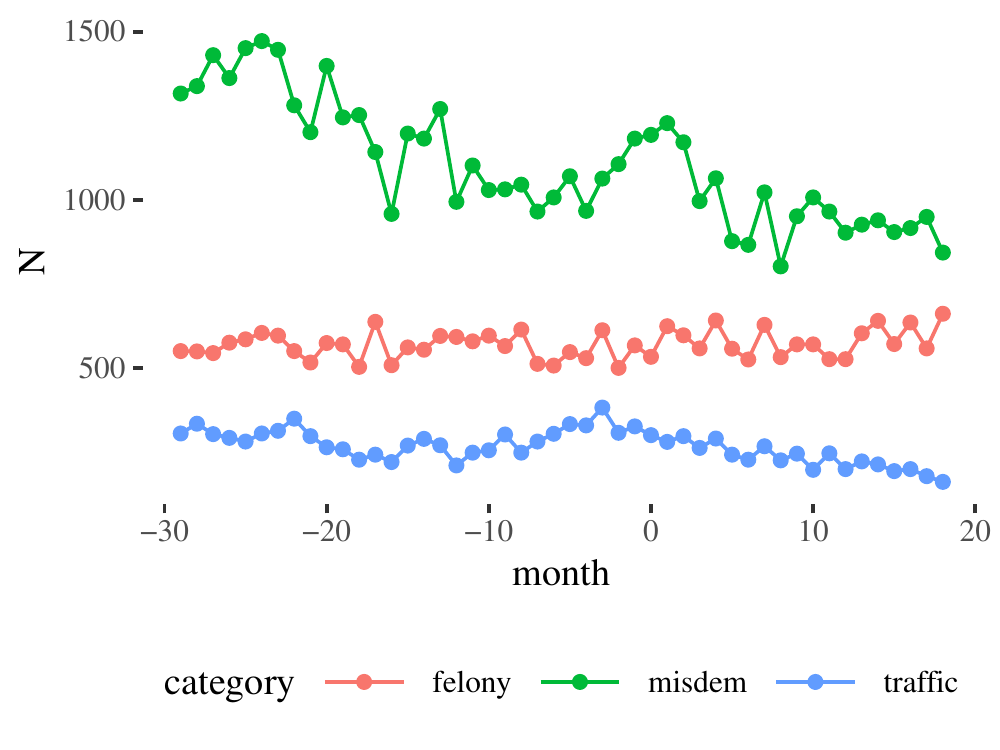}
  \caption{Number of cases of each charge type.}
  \label{fig:meck_changing_count}
\end{subfigure}
\begin{subfigure}[t]{.48\textwidth}
  %\centering
  \includegraphics[width=\textwidth]{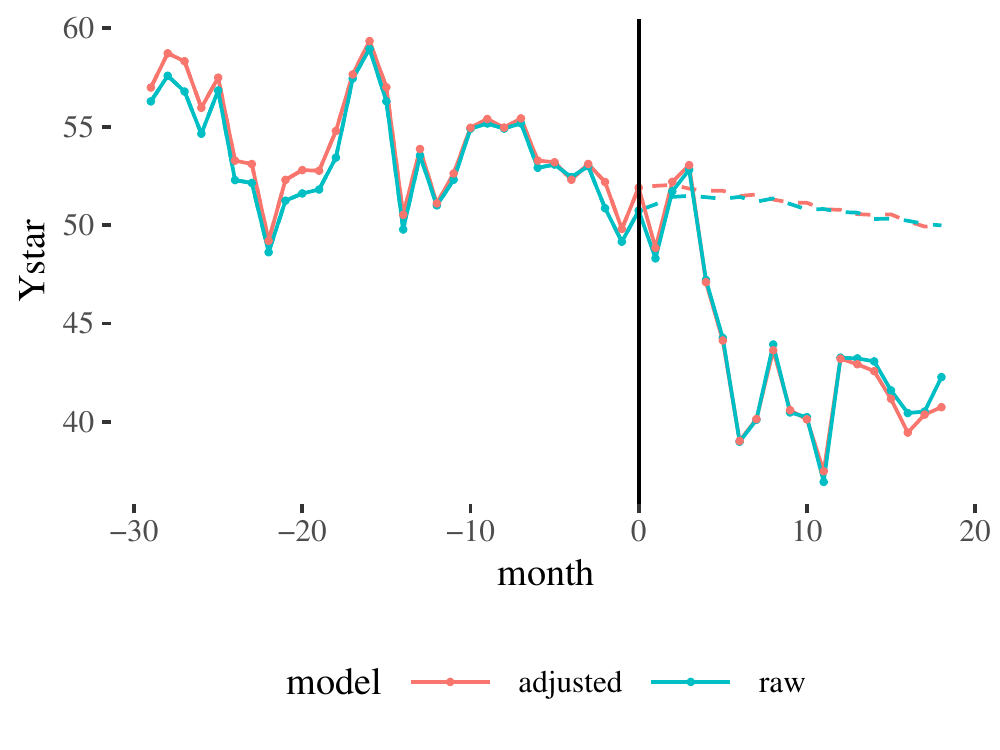}
  \caption{Adjusted impact on bail setting}
  \label{fig:meck_adjusted}
\end{subfigure}
\caption{Mecklenberg results adjusted by charge type.}{At left we see the number of misdemeanors falling at a non-linear rate; this causes the mix of cases to be changing in a complex way over time. At right our adjust rate of bail setting and associated impact analysis when we reweight each month to have a canonical mix of cases.}
\label{fig:meck_poststrat}
\end{figure}

%As further illustration, we repeat this with the count data rather than proportions.
%Here the outcome is the total number of cases given bail in each month.
%We again re-weight these totals by rescaling 

\subsection*{Case study: A sensitivity check for Mecklenberg}

We can use post-stratification to investigate whether our original results are sensitive to the changing case mix of cases in Mecklenberg over time.
To do this we first calculate what proportion of the charges were in each class post-policy (we select post-policy to capture current trends rather than historical, but this is a design decision up to the analyst).
As stated above, we find 33\% of the charges are felony, 54\% misdemeanor, and 13\% traffic.
We then reweight each month pre- and post-policy, getting an adjusted rate of bail setting as the weighted average of the rates within each of our three groups.
Finally, we use our standard lagged outcome analysis on these reweighted totals to get our predicted trend line.
The results are on Figure~\ref{fig:meck_adjusted}; note how both our counterfactual extrapolation and the observed trend are impacted by the reweighting.
That being said, we see little change in the overall impacts, and therefore our substantive results remain the same.
There is no evidence that the results are substantially driven by the change in case mix, as measured by category of case.

\section*{Appendix D: R Package overview}
We provide an R package, \verb|simITS|, to implement the methods described in this paper.
The documentation with the package is more comprehensive than this appendix, and includes a vignette walking through parts of the Mecklenberg and New Jersey analyses, but we provide a brief overview of the New Jersey analysis here as well.

To fit a seasonality model, first specify the functional form of the model:
\begin{verbatim}
my.model =  make.fit.season.model( ~ temperature + Q2 + Q3 + Q4 )	
\end{verbatim}
This creates a model that can then be fit to data; the named variables are all assumed to be in the dataset we will eventually analyze.
The following code fits and displays our model to the pre-policy data only (\verb|Y| is the outcome, stored as another column in our data), and does not include lagged covariates:
\begin{verbatim}
mod = my.model( dat = filter( newjersey, month <= 0 ), "Y", 
                 lagless = TRUE )
summary( mod )
\end{verbatim}
In the above, \verb|mod| is the result of simple linear regression, and the \verb|summary()| call will print out the estimated coefficients and overall $R^2$.

To conduct the simulation, first add the lagged covariates needed (the package will automatically extract the needed covariates given a specified model) and then make the call to process the outcome of interest:
\begin{verbatim}
newjersey = add.lagged.covariates( newjersey, 
                                   outcomename = "Y", 
                                   covariates = my.model )	
envelope = process.outcome.model( "Y", newjersey, t0 = 0, 
                                  R = 1000,
                                  summarize = TRUE, 
                                  smooth = FALSE,
                                  fit.model = my.model )
\end{verbatim}
If there are multiple outcomes, simply change the outcome name in the call above.
Loess smoothing can be done by changing the flag for \verb|smooth| to \verb|TRUE|.

We can make our plot as follows (this method uses the \verb|ggplot| plotting environment):
\begin{verbatim}                                  
plt <- make.envelope.graph( envelope, t0 = 0 )
plt
\end{verbatim}
See the package documentation for further specifications and details of the resulting objects returned from these primary method calls.

\end{document}